\documentclass{amsart}
\usepackage[foot]{amsaddr}
\newtheorem{theorem}{Theorem}
\newtheorem{lemma}{Lemma}
\newtheorem{proposition}{Proposition}
\newtheorem{corollary}{Corollary}

\theoremstyle{remark}
\newtheorem{remark}{Remark}

\usepackage[utf8]{inputenc}
\usepackage{enumerate}
\usepackage{mathrsfs}
\usepackage{amssymb}
\usepackage{mathtools}
\usepackage{graphicx}
\usepackage{hyperref}
\usepackage{soul}

\graphicspath{ {./images/} }

\newcommand{\norm}[1]{\left\lVert #1 \right\rVert}

\newcommand{\dbangle}[2]{\langle\!\langle #1, #2 \rangle\!\rangle}
\newcommand{\dangle}[1]{\langle\!\langle #1 \rangle\!\rangle}

\usepackage{stackrel}
\usepackage{color}
\usepackage{url}

\newcommand{\bx}{{ \mathbf{x} }}
\newcommand{\by}{{\mathbf{y}}}

\newcommand{\bn}{{\mathbf{n}}}

\newcommand{\bu}{{ \mathbf{u}}}

\newcommand{\bT}{\mathbf{T}}

\newcommand{\bS}{{\mathbf{S}}}

\newcommand{\cE}{\mathcal{E}}
\newcommand{\cT}{{\mathcal{T}}}
\newcommand{\cN}{{\mathcal{N}}}

\newcommand{\lb}{\label}
\newcommand{\be}{\begin{equation}}
\newcommand{\ee}{\end{equation}}
\newcommand{\bea}{\begin{eqnarray}}
\newcommand{\eea}{\end{eqnarray}}
\newcommand{\bzed}{{\bf 0}}

\newcommand{\btau}{{\mbox{\boldmath $\tau$}}}
\newcommand{\biota}{{\mbox{\boldmath $\iota$}}}
\newcommand{\bomega}{{\mbox{\boldmath $\omega$}}}
\newcommand{\barphi}{{\mbox{\boldmath $\varphi$}}}
\newcommand{\bchi}{{\mbox{\boldmath $\chi$}}}

\newcommand{\bpsi}{{\mbox{\boldmath $\psi$}}}

\newcommand{\btimes}{{\mbox{\boldmath $\,\times\,$}}}
\newcommand{\bdot}{{\mbox{\boldmath $\,\cdot\,$}}}
\newcommand{\bdots}{{\mbox{\boldmath $\,:\,$}}}
\newcommand{\grad}{{\mbox{\boldmath $\nabla$}}}
\newcommand{\ext}{\textbf{Ext}}

\begin{document}

\title[Inertial Momentum Dissipation]{Inertial Momentum Dissipation for Viscosity Solutions of Euler Equations: External Flow Around a Smooth Body}

\author[H. Quan]{Hao Quan$^1$}
\email{$^1$haoquan@jhu.edu}

\author[G. Eyink]{Gregory L. Eyink$^2$}
\address{$^1, ^2$Department of Applied Mathematics \& Statistics The Johns Hopkins University, Baltimore, MD 21218, USA}
\address{$^2$Department of Physics and Astronomy The Johns Hopkins University, Baltimore, MD 21218, USA}
\email{$^2$eyink@jhu.edu}



\date{\today}

\dedicatory{}

\begin{abstract}
We study the local balance of momentum for weak solutions of incompressible Euler equations obtained from the 
zero-viscosity limit in the presence of solid boundaries, taking as an example flow around a finite, smooth body.
We show that both viscous skin friction and wall pressure exist in the inviscid limit as distributions 
on the body surface. We define a nonlinear spatial flux of momentum toward the wall for the 
Euler solution, and show that wall friction and pressure are obtained from this momentum flux in the limit of 
vanishing distance to the wall, for the wall-parallel and wall-normal components, respectively. 
We show furthermore that the skin friction describing anomalous momentum transfer to the wall 
will vanish if velocity and pressure are bounded in a neighborhood of the wall and if also the essential supremum of wall-normal velocity within a small distance of the wall vanishes with this distance
(a precise form of the vanishing wall-normal velocity condition). In the latter case, all of the limiting 
drag arises from pressure forces acting on the body and the pressure at the body surface can be 
obtained as the limit approaching the wall of the interior pressure for the Euler solution. 
As one application of this result, we show that Lighthill's theory of vorticity generation 
at the wall is valid for the Euler solutions obtained in the inviscid limit. Further, in a companion
work, we show that the Josephson-Anderson relation for the drag, recently derived for strong Navier-Stokes 
solutions, is valid for weak Euler solutions obtained in their inviscid limit.

\end{abstract}
\maketitle

\section{Introduction}
It was proposed by Taylor as early as 1915 \cite{taylor1915eddy} that in turbulent fluid flows 
interacting with a solid boundary there may be a ``finite loss of momentum at the walls due to an 
infinitesimal viscosity'', and he suggested also an analogy with weak solutions of the fluid equations describing 
shocks. The corresponding phenomenon of ``inertial energy dissipation'' has been much investigated 
since Onsager pointed out the criticality of 1/3  H\"older singularity of the velocity field 
for such dissipation \cite{onsager1949statistical}: 
see  \cite{eyink1994energy,constantin1994onsager,duchon2000inertial,cheskidov2008energy}
for proofs of the necessity of these singularities and 
\cite{isett2018proof,buckmaster2019onsager} for proofs that such dissipative solutions exist. 
This line of investigation has been recently extended to wall-bounded turbulence by Bardos \& Titi
\cite{bardos2018onsager} and by several following works \cite{drivas2018onsager,bardos2019onsager,chen2020kato}, 
which all consider the balance of kinetic energy rather than momentum. 
However, there is a well-developed phenomenology of spatial ``momentum cascade'' in wall-bounded turbulent flows, closely 
analogous to the energy cascade through scales in the bulk of the flow away from solid boundaries
\cite{tennekes1972first,jimenez2012cascades,yang2017multifractal}. As discussed in \cite{eyink2022Aonsager}, the 
mathematical methods applied to study Onsager's dissipation anomaly due to energy cascade should apply as well to the 
spatial momentum cascade and the results of the present work have been explained at a physical level 
in  our general review \cite{eyink2024onsager}.

We perform the rigorous study here in the context of flow around a finite solid body with smooth surface, 
which was the subject of the famous paradox of d'Alembert \cite{dalembert1749theoria,dalembert1768paradoxe}. 
The type of situation we consider is illustrated in Fig.~\ref{fig1}, which shows a finite body 
$B$ and the exterior flow domain $\Omega={\mathbb R}^3\setminus B$ on which the incompressible 
Navier-Stokes equation is assumed to be satisfied 
 \begin{equation} 
    \partial_t\mathbf{u}^{\nu} + \grad\bdot(\mathbf{u}^{\nu}\otimes\mathbf{u}^{\nu} + p^{\nu}\mathbf{I}) 
    - \nu\triangle\mathbf{u}^{\nu} = 0, \;\;\; \grad\bdot\mathbf{u}^{\nu} = 0,
    \quad \bx\in \Omega \label{NS1}\end{equation}
subject to the boundary conditions     
\begin{equation} \mathbf{u}^{\nu}|_{\partial B} = {\bf 0}, \;\;\; \mathbf{u}^{\nu}\underset{|\mathbf{x}|\to\infty}{\sim}\mathbf{V}.\label{NS2}
\end{equation} 
Here the pressure $p^\nu$ is obtained from the Poisson equation with Neumann boundary 
conditions inherited from the previous equations:
\be -\triangle p^\nu=\grad\otimes\grad\bdots(\bu^\nu\otimes\bu^\nu), \ \bx\in \Omega;  \quad 
\frac{\partial p^\nu}{\partial n}= \nu \bn\bdot \triangle\mathbf{u}^{\nu},\  \bx\in \partial\Omega. 
\lb{NSpress} \ee 
where $\bn$ is the normal vector at the boundary $\partial B$ directed into the domain $\Omega.$
We shall assume in this work that $B\subset {\mathbb R}^3$ is closed and bounded, and crucially that the boundary $\partial B=\partial\Omega$ is a $C^\infty$
manifold embedded in ${\mathbb R}^3.$ It is not necessary that $B$ and 
$\partial B$ be connected (i.e. physically we may have multiple bodies).
See \cite{sohr2012navier} for a mathematical treatment 
of Navier-Stokes solutions in such unbounded domains (and even when the solid boundary is non-smooth) 
and see \cite{sueur2012kato} and references therein for discussion of the closely related problem 
of the rigid motion of the solid body $B$ through an incompressible fluid filling the complement. 
We consider this particular situation because of a new mathematical approach to the 
d'Alembert paradox  based on a Josephson-Anderson relation inspired by quantum superfluids \cite{eyink2021josephson}, which is the subject of a companion paper \cite{quan2024onsager} 
that builds upon our analysis here. However, our results in this paper apply with minor changes 
to other flows involving solid walls, including interior flows within bounding walls such as 
Poiseuille flows through pipes and channels. 

\begin{figure}
\center
\includegraphics[width=.75\textwidth]{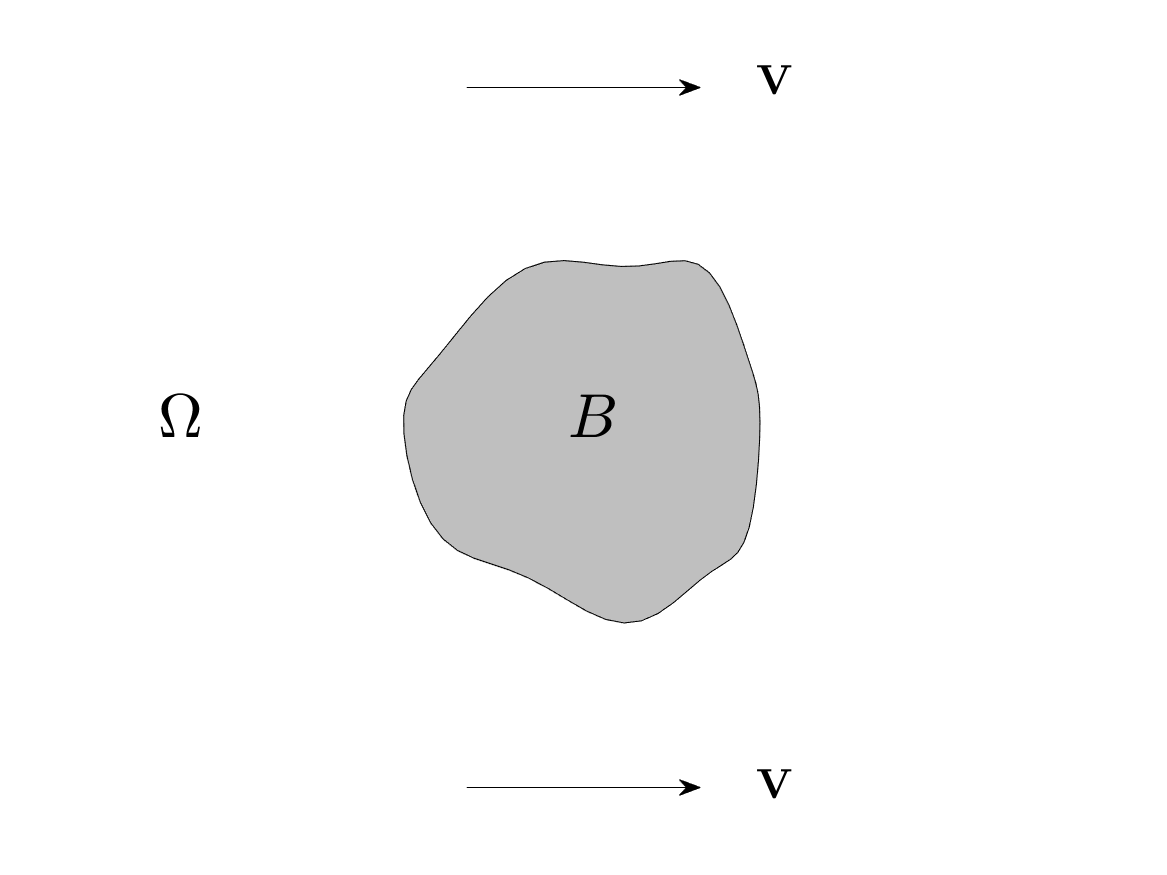}
\caption{Flow around a finite body $B$ in an unbounded region $\Omega$
filled with an incompressible fluid moving at a velocity ${\bf V}$ at far distances.} 
\label{fig1} \end{figure}

Our results and analysis here are modelled closely after those of Duchon \& Robert \cite{duchon2000inertial},
who established a kinetic energy balance distributionally in spacetime for weak solutions of 
incompressible Euler and Navier-Stokes equations. In particular, under suitable assumptions, 
\cite{duchon2000inertial}
showed that the (viscous and inertial) dissipation $\nu|\grad \bu^\nu|^2+D(\bu^\nu)$
for a sequence of Leray solutions with viscosity tending to zero must 
converge to a positive distribution (Radon measure) which agrees also with the 
inertial dissipation $D(\bu)$ for weak solutions of Euler equations obtained in the inviscid limit. 
In order to generalize the Duchon-Robert analysis to obtain a momentum balance distributionally 
in space-time, we have had to make two key modifications. First, we do not treat admissable or Leray 
weak solutions of the Navier-Stokes equations, but instead assume that all Navier-Stokes solutions are strong. The technical reason for this decision is that our argument requires consideration 
of the global momentum balance of the Navier-Stokes solution, in which spatial integration 
by parts yields an integral over $\partial B$ of the viscous Newtonian stress. 
Only recently has there been some progress in defining the trace of an averaged vorticity at the boundary for Leray solutions \cite{vasseur2023boundary}, but these results do not suffice for our analysis. There seems to be no loss of physical significance of our 
results by assuming strong solutions, however, since there is no empirical evidence 
for Leray-type singularities in any known fluid flow. The second and related difference is 
that our argument involves smearing the Navier-Stokes solutions with elements of an 
enlarged space of test functions, which need not be compactly supported in the open set 
$\Omega$ but which may instead be non-vanishing on $\partial \Omega$ and have there 
one-sided derivatives of all orders. A convenient definition of this non-standard class 
of test functions on $\bar{\Omega}\times (0,T)$ is as restrictions of standard test 
functions on ${\mathbb R}^3\times (0,T)$: 
\begin{eqnarray}
    \Bar{D}(\Bar{\Omega}\times(0,T)) &:=& \left\{\varphi = \phi|_{\Bar{\Omega}\times(0,T)}:\, 
    \phi\in C_c^{\infty}(\mathbb{R}^3\times(0,T))\right\} 
    \label{Dbar}
\end{eqnarray} 
This class of test functions is employed precisely to obtain crucial surface contributions 
from the pressure and Newtonian stress after integration by parts. As an aside, we note that for 
the initial-value problem the space $\Bar{D}(\Bar{\Omega}\times[0,T))$
$:=\left\{\varphi = \phi|_{\Bar{\Omega}\times[0,T)}:\, 
   \phi\in C_c^{\infty}(\mathbb{R}^3\times(-T,T))\right\}$
   could be similarly 
introduced, requiring slight elaboration of the arguments below.  

Our first result is that, under stated assumptions, distributional limits exist as viscosity 
tends to zero both for the normal stress or pressure and 
for the tangential Newtonian stress on the body surface, when these are considered 
as distributional sections of the normal and tangent bundles of the surface, respectively. 
More precisely, since we consider space-time distributions, we define 
the manifold 
$(\partial B)_T:=\partial B\times (0,T)\subset {\mathbb R}^3\times {\mathbb R}$  
with the natural product $C^\infty$ structure and with 
no boundary, or $\partial(\partial B)_T=\emptyset.$ Recalling that $\bn$ is the normal vector at $\partial B$ pointing into $\Omega,$ we define pressure stress acting on the wall by 
\be
-p^\nu_w \bn:= -p^\nu|_{(\partial B)_T}\bn \in D'((\partial B)_T, \mathcal{N}(\partial B)_T) 
\lb{pdist} 
\ee
as a distributional section of the normal bundle $\mathcal{N}(\partial B)_T$ and wall shear stress
\be 
    \btau^\nu_w= 2\nu\bS^\nu|_{(\partial B)_T} \bdot \bn
    =\left.\nu \frac{\partial\bu}{\partial n}\right|_{(\partial B)_T}
    =\nu\bomega^\nu|_{(\partial B)_T}\btimes\bn
    \in D'((\partial B)_T, \mathcal{T}(\partial B)_T) 
    \lb{taudist} 
\ee 
as a distributional section of the tangent bundle $\mathcal{T}(\partial B)_T.$ Here we have introduced 
the strain-rate tensor and the vorticity vector 
\be S_{ij}^\nu = \frac{1}{2}\left(\frac{\partial u_i^\nu}{\partial x_j} + \frac{\partial u_j^\nu}{\partial x_i}
\right), \quad \bomega^\nu =\grad\btimes \bu^\nu, \ee 
and note that the second equality in Eq.\eqref{taudist} is a well-known consequence of the 
stick b.c. on the velocity field \cite{lighthill1963introduction}. 
See section \ref{sec:prelim} for our notations and conventions on differential geometry. 

We then prove the following result: 

\begin{theorem}\label{theorem1} 
    Let $(\mathbf{u}^{\nu}, p^{\nu})$ be strong solutions of Navier-Stokes equations (\ref{NS1})-(\ref{NSpress}) on $\Bar{\Omega}\times (0,T)$ for $\nu>0$. 
    Assume that
    $(\mathbf{u}^{\nu})_{\nu>0}$ converges strongly to $\mathbf{u}$ in $L^2((0,T),L_{\text{loc}}^2(\Omega)):$
        \be
            \mathbf{u}^{\nu}\xrightarrow[L^2((0,T),L_{\text{loc}}^2(\Omega))]{\nu\to0}\mathbf{u}. \label{L2Conv}
        \ee    
    and that $(p^{\nu})_{\nu>0}$ converges strongly to $p$ in $L^1((0,T),L_{\text{loc}}^1(\Omega)):$    
        \be
            p^{\nu}\xrightarrow[L^1((0,T),L_{\text{loc}}^1(\Omega))]{\nu\to0}p. \label{pL1Conv}
        \ee        
    Further assume that for some $\epsilon>0$ arbitrarily small, with $\Omega_\epsilon:=\{\bx\in \Omega:\, dist(\bx,\partial B)<\epsilon\},$
        \begin{align}
            &\mathbf{u}^{\nu} \text{ uniformly bounded in } L^2((0,T),L^2(\Omega_{\epsilon}))\label{uBBound}\\
            &p^{\nu} \text{ uniformly bounded in } L^1((0,T), L^1(\Omega_{\epsilon})).\label{pBBound}
        \end{align}  
    Then, the limit $(\bu,p)$ is a weak Euler solution on $\Omega\times (0,T)$,  and $\btau^{\nu}_w,$
    $p^{\nu}_w\mathbf{n}$ have limits as surface distributions, i.e.
    \be 
        \btau^{\nu}_w \xrightarrow{\nu\to0} \btau_w \text{ in } D'((\partial B)_T, \mathcal{T}(\partial B)_T) \lb{tau-lim} 
    \ee
    \be 
        p^{\nu}_w\mathbf{n} \xrightarrow{\nu\to0} p_w\mathbf{n} \text{ in } D'((\partial B)_T, \mathcal{N}(\partial B)_T) \lb{p-lim} 
    \ee 
\end{theorem}

\begin{remark}\lb{remark1} 
This theorem is analogous to Proposition 4 of Duchon \& Robert \cite{duchon2000inertial} who proved 
that the inviscid limit of the local dissipation in Leray solutions, or 
$\lim_{\nu\to 0}[\nu|\grad \bu^\nu|^2+D(\bu^\nu)],$ exists in the sense of space-time distributions, 
under similar assumptions as ours. 
The essential identities \eqref{local_momentum},\eqref{local_momentum_n} employed 
in our proof have been previously exploited to formulate error estimates for drag and lift forces, 
for the purpose of adaptive mesh refinement in numerical simulation; see \cite{hoffman2006simulation}, Eq.(25).
The assumption \eqref{L2Conv} on strong $L^2$ convergence of velocities is motivated by results established 
and reviewed in \cite{drivas2019remarks}, which provide physically reasonable conditions for such convergence 
in the case of interior flows in bounded domains. Our assumptions \eqref{uBBound}-\eqref{pBBound} on boundedness in a small 
$\epsilon$-neighborhood of the boundary are motivated by the similar assumptions in Theorem 1 of 
\cite{drivas2018nguyen}, but are much weaker and modelled on our hypotheses  \eqref{L2Conv},\eqref{pL1Conv}. 
The latter do not, of course, imply \eqref{uBBound}-\eqref{pBBound} because the $L_{\text{loc}}^p(\Omega)$
conditions in  \eqref{L2Conv},\eqref{pBulkBound} imply boundedness of $L^p(U)$-norms 
only for $U\subset\subset\Omega.$
\end{remark}

\begin{remark} The assumption \eqref{pL1Conv} on the pressure is much stronger than required. All that 
is needed is an hypothesis which guarantees that along a suitable subsequence of $\nu$, 
$p^{\nu}\to p\in L^1((0,T), L^1_{loc}(\Omega))$ distributionally.  For example, it 
would suffice to replace \eqref{pL1Conv} instead with the following:  
    \begin{align}
        &p^{\nu} \text{ is uniformly bounded in }L^{q}((0,T),L_{\text{loc}}^{q}(\Omega)),
        \mbox{ for some $q>1$}. \label{pBulkBound}
    \end{align}
The assumption \eqref{pBulkBound} means more precisely 
that there exists an increasing sequence of open sets $\Omega_k\subset\subset\Omega_{k+1}$
with $\cup_k \Omega_k=\Omega$ such that for each $k\geq 1$
\be \sup_{\nu>0}\norm{p^\nu}_{L^q((0,T),L^q(\Omega_k))}<\infty. \ee
Thus, by the Banach-Alaoglu theorem applied iteratively in $k,$ we 
can find for each $k$ a subsequence $(\nu^{(k)})$ so that $p^{\nu^{(k)}_j}\rightharpoonup p$
weakly in $L^q((0,T),L^q(\Omega_k))$ as $j\to\infty$ and such that $(\nu^{(k+1)})$ is a 
further subsequence of $(\nu^{(k)}).$ In that case, it is easy to see that the diagonal 
subsequence $\nu_j^*=\nu_j^{(j)}$ has $\lim_{j\to\infty} p^{\nu^*_j}= p$ weakly in $L^q((0,T),L^q(\Omega_k))$ 
for all $k\geq 1,$ thus also distributionally, and then $p\in L^q((0,T),L_{loc}^q(\Omega)).$
\end{remark}

\begin{remark}\lb{remark2}
The proof of Theorem \ref{theorem1} is based on the concept of an extension operator 
for smooth test functions on the boundary into the interior flow domain. 
To prove \eqref{tau-lim} we must consider test functions $\bpsi$
on $D'((\partial B)_T, \mathcal{T}^*(\partial B)_T),$ which are smooth sections of the cotangent bundle, 
and an extension is then a map 
$\textbf{Ext}: \bpsi\in D((\partial B)_T,\mathcal{T}^*(\partial B)_T) \mapsto 
\barphi \in \Bar{D}(\Bar{\Omega}\times (0,T),\mathbb{R}^3)$ which is linear and continuous 
in the appropriate sense, with the pointwise equality  
\be \barphi|_{(\partial B)_T}= (\textbf{Proj}_{\text{s}}\circ\biota_T)(\bpsi) \lb{T-extend} \ee
where $\biota_T$ is the natural inclusion map of the tangent bundle into its ambient Euclidean space:
    \begin{align}
        \biota_T : \mathcal{T}(\partial B)_T \to (\mathbb{R}^3\times\mathbb{R})\times(\mathbb{R}^3\times\mathbb{R})
    \end{align}
and $\textbf{Proj}_{\text{s}}$ is the projection onto the spatial vector component 
\begin{align} 
    \textbf{Proj}_{\text{s}}:(\mathbb{R}^{3}\times\mathbb{R})\times(\mathbb{R}^{3}\times\mathbb{R})&\to\mathbb{R}^{3}\\
    ((\mathbf{x},t),(\bu,v))&\mapsto \bu. 
\end{align}
We define similarly the projection $\textbf{Proj}_{\text{st}}$ onto the space-time component $(\bu,v).$ 
See section \ref{sec:prelim} where we define the set 
$\mathcal{E}_{\mathcal{T}}$ of such extensions and prove that it is non-empty, by 
constructing an explicit example. Likewise, the proof of \eqref{p-lim} requires 
the definition of a set $\mathcal{E}_{\mathcal{N}}$ consisting of continuous linear extensions 
$\textbf{Ext}: \bpsi\in D((\partial B)_T,\mathcal{N}^*(\partial B)_T) \mapsto 
\barphi \in \Bar{D}(\Bar{\Omega}\times (0,T),\mathbb{R}^3)$ which satisfy the analogous
pointwise equality as \eqref{T-extend} for smooth sections of the conormal bundle.  
\end{remark} 

The weak Euler solutions obtained in Theorem \ref{theorem1} are ``viscosity solutions'' resulting 
from the inviscid limit. Weak solutions are equivalent to ``coarse-grained solutions'' in the 
sense of \cite{drivas2018onsager}, with slight modifications made due to the 
presence of boundaries. As in \cite{drivas2018onsager}, we introduce the spatial coarse-graining operation 
\be
  f\in L^1_{{\rm loc}}(\Omega)\mapsto   \Bar{f}_{\ell}(\mathbf{x}) = \int_{\mathbb{R}^3}G_\ell(\mathbf{r})f(\mathbf{x}+\mathbf{r})\, V(d\mathbf{r}), \quad
  \mathbf{x}\in\Omega^{\ell}:=\Omega\setminus \Omega_\ell
\ee
with $G_{\ell}(\mathbf{r})\coloneqq \ell^{-3}G(\mathbf{r}/\ell)$ 
a standard mollifier, assumed supported on the unit ball for simplicity. To take into account 
the domain boundary, following \cite{bardos2018onsager,drivas2018nguyen} we introduce a smooth {\it window function}  $\theta_{h,\ell}:\mathbb{R}\mapsto [0,1]$, which is non-decreasing, 0 on $(-\infty, h],$ and 1 on $[h+\ell,\infty)$, 
with derivative $\norm{\theta_{h,\ell}'}_{L^{\infty}(\mathbb{R})}\le C\ell^{-1}$ for some constant $C$ independent of $h$ and $\ell$. 
We then denote $\eta_{h,\ell}(\bx) := \theta_{h,\ell}(d(\mathbf{x}))$, where $d$ is the distance function 
\be d(\bx):= \min_{\by\in \partial B}|\bx-\by| \ee 
noting that for $\bx\in \Omega_\epsilon$ with sufficiently small $\epsilon>0,$ $d(\bx)=|\bx-\pi(\bx)|$ for a unique 
choice $\pi(\bx)\in \partial B$ and $\grad d(\bx)=\bn(\pi(\bx)):=\bn(\bx).$ See \cite{bardos2018onsager,drivas2018nguyen} and also section \ref{sec:prelim}. If the Navier-Stokes momentum balance equation \eqref{NS1} is both coarse-grained and windowed, then for 
$\ell<h$ it yields:
\begin{align}
    \partial_t(\eta_{h,\ell}\Bar{\mathbf{u}}_{\ell}^{\nu}) + \grad\bdot(\eta_{h,\ell} \Bar{\mathbf{T}}_{\ell}^{\nu} + \eta_{h,\ell}\Bar{p}_{\ell}^{\nu}\mathbf{I}) =  \grad\eta_{h,\ell}\bdot\Bar{\mathbf{T}}_{\ell}^{\nu} + \Bar{p}_{\ell}^{\nu}\grad\eta_{h,\ell} + \nu \eta_{h,\ell}\triangle\Bar{\bu}_{\ell}^{\nu} \label{cgMomentum1}
\end{align}
where we have introduced the {\it advective stress tensor} $\Bar{\mathbf{T}}_{\ell}^{\nu} = \overline{\mathbf{u}^{\nu}\otimes\mathbf{u}^{\nu}}.$
The following result describes the inviscid limit: 

\begin{proposition}\label{prop1} 
Assume conditions (\ref{L2Conv})-(\ref{pBulkBound}) as in Theorem \ref{theorem1}. Then as $\nu\to0$, the coarse-grained momentum equation (\ref{cgMomentum1}) converges pointwise for $\mathbf{x}\in\Omega$ and distributionally for $t\in[0,T]$ to the following equation,
\begin{align}
    \partial_t(\eta_{h,\ell}\Bar{\mathbf{u}}_{\ell}) 
    +\grad\bdot(\eta_{h,\ell} \Bar{\mathbf{T}}_{\ell} + \eta_{h,\ell}\Bar{p}_{\ell}\mathbf{I}) 
    = \grad\eta_{h,\ell}\bdot \Bar{\mathbf{T}}_{\ell} + \Bar{p}_{\ell}\grad\eta_{h,\ell}.\label{coarseGrainEuler}
\end{align}
with $\Bar{\mathbf{T}}_{\ell} = \overline{\mathbf{u}\otimes\mathbf{u}}$ for the limiting Euler solution $(\bu,p)$ in 
Theorem \ref{theorem1}. The set of equations \eqref{coarseGrainEuler} for all $h>\ell>0$ are equivalent to the standard 
weak formulation of the momentum balance for incompressible Euler equations. 
\end{proposition}
\noindent 
The proof of this proposition is straightforward and left to the reader. For the final statement,  
see \cite{drivas2018onsager}, Section 2. The importance of the proposition is that it identifies 
{\it nonlinear spatial flux of momentum} toward the wall at distance $h$ as 
\be -(\grad\eta_{h,\ell}\bdot \Bar{\mathbf{T}}_{\ell} + \Bar{p}_{\ell}\grad\eta_{h,\ell})
\in D'((0,T),C^\infty_c(\Omega)), \lb{mom-flux} \ee
where recall that $\grad\eta_{h,\ell}=\eta_{h,\ell}'(d(\bx))\bn(\pi(\bx)),$ when $h$ is sufficiently small. 

Our next main theorem states that this spatial flux of momentum (both its components wall-parallel and wall-normal)
matches onto the corresponding components of the limiting wall stress which were established in Theorem \ref{theorem1}. 
Since those inviscid limits were defined as sectional distributions of the tangent and normal bundles, we must 
identify momentum flux \eqref{mom-flux} with similar sectional distributions. To accomplish this, we use 
the idea of extensions in the proof of Theorem \ref{theorem1} to define e.g. 
$\mathbf{Ext}^*(\grad\eta_{h,\ell}\bdot\Bar{\mathbf{T}}_{\ell} + \Bar{p}_{\ell}\grad\eta_{h,\ell})
\in D'((\partial B)_T,{\mathcal T}(\partial B)_T)$ with $\mathbf{Ext}\in {\mathcal E}_T$ as 
$$ \langle \mathbf{Ext}^*(\grad\eta_{h,\ell}\bdot\Bar{\mathbf{T}}_{\ell} + \Bar{p}_{\ell}\grad\eta_{h,\ell}),\bpsi\rangle
= \langle \grad\eta_{h,\ell}\bdot\Bar{\mathbf{T}}_{\ell} + \Bar{p}_{\ell}\grad\eta_{h,\ell}, \mathbf{Ext}(\bpsi)\rangle$$ 
for all $\bpsi\in D((\partial B)_T,{\mathcal T}^*(\partial B)_T).$ The righthand side is meaningful and defines a 
sectional distribution of the tangent bundle because of regularity \eqref{mom-flux} and linearity and continuity 
of $\mathbf{Ext}\in {\mathcal E}_T.$  Likewise, with 
$\mathbf{Ext}\in {\mathcal E}_N$ one can define $\mathbf{Ext}^*(\grad\eta_{h,\ell}\bdot\Bar{\mathbf{T}}_{\ell} + \Bar{p}_{\ell}\grad\eta_{h,\ell}) \in D'((\partial B)_T,{\mathcal N}(\partial B)_T).$ 
For details, see Section \ref{sec:prelim}.

We then have: 
\begin{theorem}\label{theorem2} 
    Under the assumptions \eqref{L2Conv}-\eqref{pBBound} of Theorem 1, then for 
    $0<\ell<h$ and for all $\textbf{Ext}\in \mathcal{E}_T$          
    \be   - \lim_{h,\ell\to0}\textbf{Ext}^*(\grad\eta_{h,\ell}\bdot\Bar{\mathbf{T}}_{\ell} + \Bar{p}_{\ell}\grad\eta_{h,\ell}) = \btau_w \mbox{ in }D'((\partial B)_T, \mathcal{T}(\partial B)_T)\label{tangentialLimit} \ee
Likewise, for $0<\ell<h$ and for all $\textbf{Ext}\in {\mathcal E}_N$ 
    \be
        -\lim_{h,\ell\to0} \textbf{Ext}^*(\grad\eta_{h,\ell}\bdot\Bar{\mathbf{T}}_{\ell} + \Bar{p}_{\ell}\grad\eta_{h,\ell}) = -p_w\mathbf{n} \mbox{ in }D'((\partial B)_T, \mathcal{N}(\partial B)_T). \label{normalLimit} \ee
\end{theorem}

\begin{remark}
This result is analogous to the second part of  Proposition 4 of Duchon \& Robert \cite{duchon2000inertial}, stating not only that $D(\bu)=\lim_{\nu\to 0}[\nu|\grad \bu^\nu|^2+D(\bu^\nu)]$ exists but also that it coincides with the ``inertial energy dissipation'' of \cite{duchon2000inertial}, Proposition 2, which 
defines it as a distributional limit of energy flux to vanishingly small length scales,
$D(\bu)=\lim_{\ell\to 0}D_\ell(\bu).$ In fact, our proof of Theorem \ref{theorem2} 
is a direct adaptation of the proof in \cite{duchon2000inertial}. 
\end{remark}

\begin{remark}
It is not geometrically natural that pressure stress should contribute to the cascade of wall-parallel 
momentum, as it apparently does in \eqref{tangentialLimit}. In fact, as previously noted, $\grad\eta_{h,\ell}
=\theta_{h,\ell}'\bn$ for sufficiently small $h,$ and the term $\Bar{p}_{\ell}\grad\eta_{h,\ell}$ 
should give vanishing contribution in the tangent bundle. This can be shown if we define a class 
of {\it natural extensions} $\tilde{\cE}_{\mathcal{T}}$ which consists of those $\textbf{Ext}\in\mathcal{E}_{\mathcal{T}}$ 
such that $\forall\bpsi\in D((\partial B)_T, \mathcal{T}^*(\partial B)_T)$, $\barphi = \textbf{Ext}(\bpsi)$ 
satisfies 
   \begin{align}
        \norm{\barphi\bdot\mathbf{n}}_{L^{\infty}((\Omega_{h+\ell}\backslash\Omega_h)\times(0,T))}\leq C\ell \label{extCond1}
    \end{align}
(possibly with $\ell/h$ bounded from below) for constant $C$ independent of $h,$ $\ell.$    
We show in Section \ref{sec:prelim} that $\tilde{\cE}_{\mathcal{T}}\neq \emptyset$ by explicit construction. 
We then obtain from the preceding theorem the following simple corollary: 
\end{remark}

\begin{corollary}\label{corollary1} 
   For $\textbf{Ext}\in\tilde{\cE}_{\mathcal{T}},$ then under the assumption 
   \eqref{pBBound} of Theorem 1, $\lim_{h,\ell\to0}\textbf{Ext}^*(\Bar{p}_{\ell}\grad\eta_{h,\ell}) = 0.$ 
   Thus, under all of the assumptions \eqref{L2Conv}-\eqref{pBBound} of Theorem 1, 
   \be 
        -\lim_{h,\ell\to0} \textbf{Ext}^*(\grad\eta_{h,\ell}\bdot\Bar{\mathbf{T}}_{\ell}) = \btau_w\text{ in }D'((\partial B)_T, \mathcal{T}(\partial B)_T)
    \ee 
    for any $\textbf{Ext}\in\tilde{\cE}_{\mathcal{T}}.$
\end{corollary}

Finally, we establish sufficient conditions for vanishing cascade of momentum to the wall via spatial advection:
\begin{proposition}\label{proposition2} 
    Assume that $\mathbf{u} \in L^2((0,T),L^2_{{\rm loc}}(\Omega))$ so that $\bT_\ell=\overline{\bu\otimes\bu}$
    is well-defined. Assume further for some $\epsilon>0$ the boundedness property in the vicinity of the wall
        \begin{align}
            \mathbf{u} \in L^2((0,T),L^{\infty}(\Omega_{\epsilon}))\label{wallboundedness}
        \end{align}
    and vanishing wall-normal velocity at the boundary in the sense
        \begin{align}
            \lim_{\delta\to0}\norm{\mathbf{n}\bdot\mathbf{u}}_{L^2((0,T),L^{\infty}(\Omega_{\delta}))} = 0. \label{wallnormal}
        \end{align}
     Then, for all $\textbf{Ext}\in\cE_{\mathcal{T}},$
         \be 
         \lim_{h,\ell\to 0} \textbf{Ext}^*(\grad\eta_{h,\ell}\bdot\Bar{\mathbf{T}}_{\ell}) = \bzed \mbox{ in } D'((\partial B)_T, \mathcal{T}(\partial B)_T)
         \lb{prop1a} 
         \ee 
     and for all $\textbf{Ext}\in {\mathcal E}_N,$
         \be 
         \lim_{h,\ell\to 0} \textbf{Ext}^*(\grad\eta_{h,\ell}\bdot\Bar{\mathbf{T}}_{\ell}) = \bzed, \mbox{ in } D'((\partial B)_T, \mathcal{N}(\partial B)_T). 
         \lb{prop1b} 
         \ee 
\end{proposition}     
     
\begin{remark}
This result can be regarded as an analogue of Duchon \& Robert, \cite{duchon2000inertial} Proposition 3,
which showed that $\lim_{\ell\to 0}D_\ell(\bu)=0$ when the velocity field satisfies a regularity condition 
slightly stronger than $\bu\in L^3((0,T),B^{1/3,\infty}_3(\Omega)).$ Our assumption \eqref{wallnormal}
can be regarded as a corresponding assumption on continuity of the normal velocity at the wall, the 
importance of which has been recognized in prior work: see Remark 3.2 in \cite{bardos2018onsager}, assumption 1, 
Eq.(4.3b) of Theorem 4.1 in \cite{bardos2019onsager}, and assumption (12) of Theorem 1 in \cite{drivas2018nguyen}.
Our near-wall boundedness assumption \eqref{wallboundedness} is likewise 
motivated by assumption (11) of Theorem 1 
in \cite{drivas2018nguyen}, but requiring only $L^2$ rather than $L^3$ sense in time. 
\end{remark}

Combining Proposition \ref{proposition2} with Theorems \ref{theorem1} \& \ref{theorem2}, and 
Corollary \ref{corollary1} yields our main result:      
     
\begin{theorem}\label{theorem3}       
Make all of the assumptions \eqref{L2Conv}-\eqref{pBBound}  of Theorem 1, and assume further 
that the limiting weak Euler solution $(\bu,p)$ in that theorem satisfies the near-wall 
boundedness \eqref{wallboundedness} and vanishing wall-normal velocity \eqref{wallnormal} in Proposition \ref{proposition2}. 
Then, for all $\textbf{Ext}\in \tilde{\cE}_{\mathcal{T}},$
\be
-\lim_{h,\ell\to 0} \textbf{Ext}^*(\grad\eta_{h,\ell}\bdot\Bar{\mathbf{T}}_{\ell}) 
\,= \btau_w\,= \bzed \mbox{ in } D'((\partial B)_T, \mathcal{T}(\partial B)_T)
\lb{tauzed} 
\ee 
and for all $\textbf{Ext}\in {\mathcal E}_N,$
\be 
-\lim_{h,\ell\to 0} \textbf{Ext}^*(\Bar{p}_{\ell}\grad\eta_{h,\ell}) = -p_w\mathbf{n}, \mbox{ in } D'((\partial B)_T, \mathcal{N}(\partial B)_T). 
\lb{pcont} 
\ee 
where the distributions $\btau_w\in D'((\partial B)_T, \mathcal{T}(\partial B)_T),$ 
$p_w\bn\in D'((\partial B)_T, \mathcal{N}(\partial B)_T)$ are those obtained in Theorem \ref{theorem1}. 
\end{theorem}

\begin{remark}
The result \eqref{tauzed} implies that Taylor's conservation anomaly for tangential momentum, under the stated 
hypotheses, can be at most a ``weak anomaly''. Here we employ the terminology from \cite{bedrossian2019sufficient}
(also \cite{eyink2022Aonsager})  according to which $\btau_w^\nu$ is ``weakly anomalous'' if it vanishes 
as $Re\to\infty,$ but more slowly than it does for laminar flow where $\btau_w^\nu \propto 1/Re.$ Such a 
weak anomaly for tangential momentum conservation would imply that all drag in the inviscid limit arises 
from the ``form drag'' due to pressure stress \eqref{pcont} acting in the direction of the external flow $\mathbf{V}.$

There is a good deal of empirical evidence from experiments and numerical simulations 
which supports this picture. For example, in the experimental study \cite{achenbach1972experiments}
for high-Reynolds flow around a smooth sphere, $\btau_w^\nu\propto Re^{-1/2}$ in the front of the 
sphere, consistent with the boundary-layer theory of Prandtl \cite{prandtl1905flussigkeitsbewegung,
lighthill1963introduction,e2000boundary}, and vanishes a bit slower in the turbulent wake region 
after flow separation behind the sphere (see \cite{achenbach1972experiments},Fig.7(a)). The form drag 
from pressure stress thus becomes becomes dominant for very large Reynolds numbers (see \cite{achenbach1972experiments}, Fig.10). 
For flow through a straight, smooth-walled pipe, as reviewed in \cite{eyink2022Aonsager}, geometry does 
not permit wall pressure stress to act parallel to the mean flow direction 
and drag vanishes as $Re\to\infty.$   If, instead, the pipe walls are mathematically smooth but ``hydraullically rough'', 
then  form drag is again geometrically possible and it becomes dominant over the contribution 
from $\btau_w^\nu$ in the large-$Re$ limit; e.g. see \cite{busse2017reynolds}, Fig.10. 
For related evidence in many other flows, see \cite{cadot1997energy,eyink2021josephson}.

The only possible exception of which we are aware comes from a 2D numerical simulation 
of a vortex quadrupole impinging on a flat wall \cite{nguyenvanyen2018energy}. Evidence was 
presented in \cite{nguyenvanyen2018energy}, Figure 12, that the maximum vorticity at 
the wall in that flow scales $\sim Re,$ which would imply $\btau_w\neq \bzed$ at least at one point.
It is possible that our strong version \eqref{wallnormal} of the vanishing wall-normal velocity 
is invalid in this flow, since reference \cite{nguyenvanyen2018energy} reports ``a blow-up of the wall-normal 
velocity associated with an abrupt acceleration of fluid particles away from the wall,'' 
corresponding to explosive boundary-layer separation. Another possible reconciliation of our 
Theorem \ref{theorem3} with the numerical observations of \cite{nguyenvanyen2018energy} is that 
the nonzero $\btau_w$ values reported may occur as $\nu\to 0$ at only a zero-measure set of points of 
$\partial \Omega,$ so that still $\btau_w=\bzed$ in the sense of distributions
and $\lim_{\delta\to 0}{\rm ess.sup}_{\bx\in \Omega_\delta}|\bn\bdot\bu(\bx)|=0.$ 
\end{remark} 

\vspace{-15pt} 

\begin{remark}
On the other hand, the assumptions \eqref{wallboundedness}, \eqref{wallnormal} invoked 
in Theorem \ref{theorem3} imply the weak-strong uniqueness property for the resulting 
viscosity solutions of Euler equations, e.g. see \cite{wiedemann2018weak}. 
(We thank T. Drivas for insisting on this fact.) This result is immediate when the 
flow domain $\Omega$ is a bounded open set with $C^\infty$ boundary $\partial \Omega$
and if there is an incompressible Euler solution ${\bf U}\in C^\infty(\Omega\times [0,T)$) 
which satisfies ${\bf U}\bdot\bn=0$ everywhere on the boundary. In that case, we may consider 
${\bf U}$ as an extension $\barphi$ into $\Omega$ of a smooth section of the surface cotangent bundle
and from the proof of Theorem  \ref{theorem3} we obtain that the limiting viscosity 
solution $\bu$ must satisfy for a.e. $\tau\in (0,T)$
\begin{eqnarray}
\int_\Omega \left[\bu(\cdot,\tau)\bdot{\bf U}(\cdot,\tau)-\bu_0\bdot{\bf U}_0\right]\, dV   
=\int_0^\tau \int_\Omega [\partial_t{\bf U}\bdot\bu+\grad{\bf U}\bdots\bu\otimes\bu]\,dV\,dt.
\end{eqnarray}
weak-strong uniqueness for the admissable weak solution $\bu$ then follows by a remark of 
E. Feireisl recorded in \cite{wiedemann2018weak}, section 5. This argument may not apply
if ${\bf U}\bdot\bn\neq 0$ on part of the boundary (as for open flows through pipes), since 
the above equation then gets a surface contribution from the pressure $p$ of the weak solution.
This argument also does not apply for flow around a smooth finite body $B$ as discussed 
in the present paper, because the smooth Euler solution ${\bf U}$ will not generally be compactly supported 
in $\Omega$ and cannot be regarded as a smooth extension. However, we show in our companion 
paper \cite{quan2024onsager} that weak-strong uniqueness nevertheless holds by a relative energy 
argument when ${\bf U}$ is the potential Euler solution of d'Alembert and when assumptions 
\eqref{wallboundedness}, \eqref{wallnormal} of Theorem \ref{theorem3} hold $L^3$-in-time.
In particular, if initial data $\bu^\nu_0$ for the Navier-Stokes solution converges to ${\bf U}_0$
strong in $L^2(\Omega)$ (allowing a vanishing boundary layer to enforce stick conditions 
at the surface), then the limiting weak Euler solution $\bu$ must coincide with ${\bf U},$
unless the conditions \eqref{wallboundedness}, \eqref{wallnormal} are violated. It should
be emphasized that, in fact, it is the consequence $\btau_w=\bzed$ of Theorem \ref{theorem3}
which implies weak-strong uniqueness for viscosity weak solutions, even if $\btau_w=\bzed$ 
follows from assumptions weaker than \eqref{wallboundedness}, \eqref{wallnormal}. This 
statement agrees with a general result of Bardos \& Titi (\cite{bardos2013mathematics}, Theorem 4.1) which implies weak-strong uniqueness under the same condition $\btau_w=\bzed$ even for 
weak-* limits in $L^\infty((0,T),L^2(\Omega))$ of Navier-Stokes solutions $\bu^\nu.$ 
Since $\btau^\nu_w=\nu\bomega^\nu_w\btimes\bn,$ a thin enough boundary layer in the initial 
data $\bu^\nu_0$ may correspond to $\btau^\nu_w\sim O(1)$ in the surface vortex sheet and 
subsequent explosive separation of such a boundary layer may violate our assumptions \eqref{wallboundedness}, \eqref{wallnormal}
at early times. 
\end{remark}

\vspace{-5pt} 
\begin{remark} 
The result \eqref{pcont} of Theorem \ref{theorem3} is a statement that pressure is continuous  
at the wall in the inviscid limit, in the sense that the limit of zero distance to the wall and the limit of 
infinite Reynolds-number commute with each other. Such continuity helps to justify one of the 
fundamental assumptions in the theory of Prandtl \cite{prandtl1905flussigkeitsbewegung,
lighthill1963introduction,e2000boundary}, which posited that pressure would be continuous
across thin viscous boundary layers at solid walls.

This result has further important implications for turbulence modelling, because it suggests that the 
asymptotic drag arising from pressure forces might be calculated from Euler solutions in the fluid 
interior which arise from the infinite-$Re$ limit \cite{drivas2019remarks}, without the need 
to resolve small viscous lengths at the wall. To obtain the pressure field 
$p$ from the Euler solution velocity field $\bu$ involves the solution of a Poisson equation 
analogous to Eq.\eqref{NSpress}, and this requires suitable boundary conditions on the pressure. 
For smooth Euler solutions, the following Neumann problem is generally solved: 
\be -\triangle p=\grad\otimes\grad\bdots(\bu\otimes\bu), \ \bx\in \Omega;  \quad 
\frac{\partial p}{\partial n}= (\bu\otimes\bu)\bdots\grad\bn,\  \bx\in \partial\Omega, 
\lb{Epress} \ee 
where the latter condition arises from the normal component of the Euler equation at the wall, 
assuming $\bu\bdot\bn=0.$
Recently, in interesting work \cite{derosa2023double,derosa2024full} (following \cite{bardos2022c})
it has been shown, assuming a weak Euler solution in a bounded domain $\Omega$ with velocity 
$\bu\in C^\alpha(\Omega),$ $\alpha\in (0,1)$ and $\bu\bdot\bn=0$ on $\partial\Omega,$ that 
the pressure $p$ must satisfy the Neumann problem \eqref{Epress} in the weak form 
\be \int_\Omega [p\triangle\varphi+\bu\otimes\bu\bdots (\grad\otimes\grad)\varphi]\,dV 
=\int_{\partial \Omega} p \frac{\partial\varphi}{\partial n}\, dA, 
\quad\forall \varphi\in C^2(\Bar{\Omega}) \lb{Epress-wk} \ee
and, furthermore, that  there is a unique weak solution 
of \eqref{Epress-wk} with zero space-mean which is at least $C^{2\alpha}$ up to the boundary. 
This result offers hope that the drag on the body in the infinite Reynolds limit 
can be computed entirely from the limiting weak Euler solution.  
\end{remark} 

\begin{remark}
The methods of this paper can be applied to another fundamental cascade process in 
wall-bounded turbulence, which is the ``inverse cascade'' of vorticity away from the wall; 
e.g. see \cite{eyink2008turbulent,eyink2021josephson}. 
Here we just note a key result for inviscid-limit  
Euler solutions which follows directly from the considerations in the present paper: 
with the assumptions of Theorem \ref{theorem3},
then for all $\mathbf{Ext}\in {\mathcal E}_{\mathcal T},$ 
\be \lim_{h,\ell\to 0}
\mathbf{Ext}^*\left[\grad\eta_{\ell,h}\btimes\partial_t\Bar{\bu}_\ell + 
\grad\eta_{h,\ell}\btimes (\grad\bdot\Bar{\bT}_\ell)\right] 
=-(\bn\btimes\grad)p_w. \lb{vort-bal} \ee
The quantity on the righthand side of this equation is the {\it Lighthill vorticity source} 
\cite{lighthill1963introduction,morton1984generation,eyink2021josephson}, which describes the rate of generation 
of tangential vorticity due to pressure gradients at the body surface. The term involving 
$\Bar{\mathbf{T}}_\ell$ on the lefthand side represents a  spatial flux of vorticity away from the 
solid surface; e.g. see \cite{eyink2024onsager}. 
One might naively expect the Lighthill source 
to be in balance with this vorticity flux into the flow interior $\Omega.$ However, the time-derivative 
term has also a simple physical interpretation, representing the rate of change of a 
tangential vortex sheet of strength $\bn\btimes\bu$ at the body surface $\partial \Omega$ 
\cite{eyink2024onsager}. 
The meaning of \eqref{vort-bal} is thus that vorticity generated at the surface by pressure gradients 
is either cascaded into the flow interior or else accumulates in the surface vortex sheet. 

It is worth sketching here at least briefly the derivation of this result. For any 
$\bpsi\in D((\partial B)_T,{\mathcal T}^*(\partial B)_T),$ let $\barphi=\mathbf{Ext}(\bpsi).$ 
Then it is not hard to show that $((\bn\btimes\grad)\bdot\bpsi)\bn\in D((\partial B)_T,{\mathcal N}^*(\partial B)_T)$
and that $(\bn\bdot(\grad\btimes\barphi))\bn\in \Bar{D}(\Bar{\Omega}\times (0,T),{\mathbb R}^3)$ 
extends this test section into the interior. Since 
\be -\langle (\bn\btimes\grad)p_w,\bpsi\rangle= \langle p_w\bn, ((\bn\btimes\grad)\bdot\bpsi)\bn\rangle \ee 
we obtain from \eqref{pcont} in Theorem \ref{theorem3} that 
\be -\langle (\bn\btimes\grad)p_w,\bpsi\rangle=\lim_{h,\ell\to 0} \int_0^T\int_\Omega
(\grad\btimes\barphi)\bdot\grad\eta_{h,\ell}\,\Bar{p}_\ell\,dV\,dt. \lb{step1} \ee 
On the other hand,
\begin{eqnarray}
&& \int_0^T\int_\Omega
(\grad\btimes\barphi)\bdot\grad\eta_{h,\ell}\,\Bar{p}_\ell\,dV\,dt
= -\int_0^T\int_\Omega
\barphi\bdot\grad\eta_{h,\ell}\btimes\grad\Bar{p}_\ell\,dV\,dt \cr
&& \hspace{10pt} = \int_0^T\int_\Omega\left[
-(\partial_t\barphi)\bdot\grad\eta_{h,\ell}\btimes\Bar{\bu}_\ell
+\barphi\bdot\grad\eta_{h,\ell}\btimes(\grad\Bar{\bT}_\ell)\right]\,dV\,dt 
\lb{step2} \end{eqnarray}
where in the final line we used the coarse-grained momentum balance 
\eqref{coarseGrainEuler}. Combining the two results \eqref{step1},\eqref{step2} yields exactly 
\eqref{vort-bal}, thus showing that the Lighthill theory of vorticity generation is valid 
even in the infinite Reynolds-number limit. The inviscid nature of vorticity production 
by tangential pressure gradients was already emphasized by Morton \cite{morton1984generation}. 
\end{remark} 

\begin{remark} 
A further application of the results of this work is given in the companion paper \cite{quan2024onsager}, 
where the infinite-Reynolds limit is
established for the Josephson-Anderson relation, which 
precisely relates vorticity flux from the body to drag \cite{eyink2021josephson}. That relation 
decomposes the velocity into a contribution $\bu_\phi=\grad\phi$ from the smooth, potential
Euler solution studied by d'Alembert \cite{dalembert1749theoria,dalembert1749theoria}
and the complementary contribution $\bu_\omega^\nu=\bu^\nu-\bu_\phi$ which represents 
the rotational fluid motions. Most importantly, this field satisfies an equation 
for conservation of ``rotational momentum" 
 \begin{equation} 
    \partial_t\mathbf{u}^{\nu}_\omega + \grad\bdot(\mathbf{u}^{\nu}_\omega\otimes\mathbf{u}^{\nu}_\omega 
    +\mathbf{u}^{\nu}_\omega\otimes\mathbf{u}_\phi+\mathbf{u}_\phi\otimes\mathbf{u}^{\nu}_\omega  
    + p^{\nu}_\omega\mathbf{I}) 
    - \nu\triangle\mathbf{u}^{\nu}_\omega = \bzed, \;\;\; \grad\bdot\mathbf{u}^{\nu}_\omega = 0,
    \quad \bx\in \Omega \label{NS1-om}
\end{equation}
subject to the boundary conditions     
\begin{equation} 
\mathbf{u}^{\nu}_\omega|_{\partial B} = 
 -\mathbf{u}_\phi|_{\partial B} , \;\;\; \mathbf{u}^{\nu}_\omega\underset{|\mathbf{x}|\to\infty}{\sim}\mathbf{0}.\label{NS2-om}
\end{equation} 
and with the pressure $p^\nu_\omega$ determined by the incompressibility constraint.  Of course, 
Eqs.\eqref{NS1-om},\eqref{NS2-om} are equivalent to the incompressible Navier-Stokes equations 
in their standard representation, Eqs.\eqref{NS1},\eqref{NS2}. Because the equations \eqref{NS1-om}
are conservation-type, they have a weak formulation and therefore all of the results of the present 
work are valid also for Eqs.\eqref{NS1-om},\eqref{NS2-om} and, in particular, the Theorems 
\ref{theorem1}-\ref{theorem3}. Note in this context that the weak Euler solutions obtained 
in the inviscid limit satisfy in distributional sense the equations 
 \begin{equation} 
    \partial_t\mathbf{u}_\omega + \grad\bdot(\mathbf{u}_\omega\otimes\mathbf{u}_\omega 
    +\mathbf{u}_\omega\otimes\mathbf{u}_\phi+\mathbf{u}_\phi\otimes\mathbf{u}_\omega  
    + p_\omega\mathbf{I}) 
    = \bzed, \;\;\; \grad\bdot\mathbf{u}_\omega = 0. 
    \quad \bx\in \Omega \label{E1-om}
\end{equation}
The resulting weak solutions $\bu=\bu_\phi+\bu_\omega$ of incompressible Euler equations 
in their standard form differ from the potential solution $\bu_\phi$ of d'Alembert, with 
non-vanishing vorticity corresponding to the rotational flow $\bu_\omega$ in the turbulent 
wake behind the solid body.
\end{remark}

\section{Preliminaries}\label{sec:prelim} 
In this section, we summarize our notations and conventions on differential geometry 
and introduce the concept of extensions that we employ in our proofs. 

\subsection{Manifolds and Vector Bundles Associated to a Smooth Body}\label{subsec:manifold} 
We consider a body $B$ that is a connected, compact domain in $\mathbb{R}^3,$ with 
$\Omega={\mathbb R}^3\backslash B$ also connected, and with common $C^{\infty}$ boundary 
$\partial B=\partial \Omega$. 
The boundary $\partial B$ is then a connected compact $C^{\infty}$ hypersurface in $\mathbb{R}^3$, 
which is thus a level set of a  $C^{\infty}$ function $f:B\to[0,\infty)$. That is, $\partial B = f^{-1}(0)$ and $\grad f(\bx)\ne\bzed$ for all $\bx\in\partial B$. 
By the Regular Level Set Theorem (\cite{lee2013smooth}, Corollary 5.14) the tangent space at any $\bx \in \partial B$ is given by 
\begin{align}
    \mathcal{T}_{\mathbf{x}}\partial B = \text{ker}(\grad f(\bx)) = (\grad f(\bx)\mathbb{R})^{\perp}. 
\end{align}
Furthermore, the vector field
\begin{align}
    \bn(\bx) = \frac{\grad f(\bx)}{|\grad f(\bx)|}
\end{align}
defines a unit normal vector of $\partial B$, and $\bn$ is also smooth on $\partial B$ by definition. See chapter 5 in \cite{lee2013smooth} for more details on submanifolds with a boundary. \\

Since $\partial \Omega$ is a compact $C^{\infty}$ submanifold of $\Omega$, there exists $\eta(\Omega)>0$ 
such that $\Omega_{\epsilon}$ for any $\epsilon<\eta(\Omega)$ is a neighborhood of $\partial B\subset \Omega$ with 
the \textit{unique nearest point property}:
for any $\bx\in\bar{\Omega}_{\epsilon}$, there exists a unique point $\pi(\bx)\in\partial B$ such that
${\rm dist}(\bx,\partial B) = |\bx - \pi(\bx)|$. The map $\pi:\Bar{\Omega}_{\epsilon}\to\partial B$ is called the \textit{projection} onto $\partial B$. One can show this projection map $\pi$ is $C^{\infty}$ using the Tubular Neighborhood Theorem. See chapter 6 in \cite{lee2013smooth}, and \cite{foote1984regularity, leobacher2021existence} for more details. Thus the distance function $d:\bar{\Omega}_{\epsilon}\to\mathbb{R}$ is a smooth function in $C^{\infty}(\bar{\Omega}_{\epsilon})$, and
\begin{align}
    d(\bx) = {\rm dist}(\bx,\partial B) = |\bx - \pi(\bx)|, \;\;\;\; \grad d(\bx) = \bn(\pi(\bx))
\end{align}
The latter result follows by using appropriate local coordinates: see \cite{guillemin1990geometric}. p.9. 
Finally, we observe that $\partial B$ is naturally Riemannian, with metric induced by the embedding
in Euclidean space. 

We need to consider also additional manifolds associated with $B.$ The first is the space-time 
manifold $(\partial B)_T := \partial B\times (0,T)$ with the product differentiable structure,
so that $\partial(\partial B)_T=\emptyset.$ Since $(\partial B)_T$ is a closed smooth hypersurface 
in ${\mathbb R}^3\times {\mathbb R},$ it is orientable and Riemannian. We consider also the associated 
{\it tangent bundle} ${\mathcal T}(\partial B)_T$  (\cite{dieudonne1972treatise}, 16.15.4; 
\cite{tomdieck2010algebraic}, section 15.6; \cite{lee2013smooth}, Proposition 3.18).  
As $(\partial B)_T$ is an embedded submanifold of $\mathbb{R}^3\times\mathbb{R}$, $\mathcal{T}(\partial B)_T\subset(\mathbb{R}^3\times\mathbb{R})\times(\mathbb{R}^3\times\mathbb{R})$
We can describe the tangent space $T_{(\mathbf{x},t)}(\partial B)_T\cong \mathcal{T}_{\mathbf{x}}\partial B\times T_t (0,T)$ embedded in $\mathbb{R}^3\times\mathbb{R}$. 
We use $\biota_T$ to denote the natural inclusion map of the tangent bundle into its ambient Euclidean space:
    \begin{align}
        \biota_T: \mathcal{T}(\partial B)_T \to (\mathbb{R}^3\times\mathbb{R})\times(\mathbb{R}^3\times\mathbb{R}). 
    \end{align}
Finally, we need the {\it normal bundle} ${\mathcal N}(\partial B)_T$ (\cite{tomdieck2010algebraic}, section 15.6; 
\cite{lee2013smooth}, Proposition 13.21), and we can take the normal space $N_{(\mathbf{x},t)}(\partial B)_T\cong \mathcal{N}_{\mathbf{x}}\partial B\times \{0\}$ embedded in $\mathbb{R}^{3}\times\mathbb{R}$.   
We use $\biota_N$ to denote the natural inclusion map of the normal bundle into its ambient Euclidean space:
    \begin{align}
        \biota_N: \mathcal{N}(\partial B)_T \to (\mathbb{R}^3\times\mathbb{R})\times(\mathbb{R}^3\times\mathbb{R}). 
    \end{align}
Because $(\partial B)_T$ is orientable, the normal bundle ${\mathcal N}(\partial B)_T$ is trivial 
(\cite{lee2013smooth},Exercise 15.8) and every smooth section $\sigma: (\partial B)_T\to {\mathcal N}(\partial B)_T$ 
can be identified with the map $(\bx,t)\mapsto (\bx,t,\sigma(\bx,t)\bn(\bx),0)$ for a smooth function 
$\sigma:(\partial B)_T\to {\mathbb R}.$

\subsection{Distributions on Manifolds}

The results on distributions that we require in this paper follow as a special case of general theory for 
a $C^{\infty}$ manifold $X$ of dimension $n$, 
and a rank $k$ vector bundle $(E, \Pi, X)$ of $X$. 
Let $\cup_{i\in I}(V_{i}, \Phi_{i}),$ $V_i\subset X,$ $\Phi_i: \Pi^{-1}(V_i)\to {\mathbb R}^n\times {\mathbb R}^k$ 
be a smooth structure of $E$, and $\cup_{i\in I}(V_{i}, \phi_{i}),$ $\phi_i: V_i\to {\mathbb R}^n$ 
be a corresponding smooth structure on $X$ with $\Pi_1\Phi_i=\phi_i\Pi.$ Here $\Pi_1$ projects 
onto the first factor of ${\mathbb R}^n\times {\mathbb R}^k$ and $\Pi_2$ onto the second. 
We shall denote by $D(X,E)$ the space of \textit{smooth sections} with compact support, 
which is a Fréchet space with the seminorms  defined by 
\begin{equation}
    p_{s,m,i}(\psi) := \sum_{j=1}^{k}\tilde{p}_{s,m,i}((\Pi_2\Phi_{i})^j\circ \psi|_{V_{i}}\circ\phi_{i}^{-1})\label{seminorms}
\end{equation}
where $\psi\in D(X,E)$ and the $\tilde{p}_{s,m,i}$'s are a countable and separating basis 
of seminorms on $C^{\infty}(\phi_{i}(V_{i}))$ defined by 
\begin{align}
    \tilde{p}_{s, m, i}(f) = \sup_{x\in K_m^{(i)}, |\alpha|\le s}|D^{\alpha}f(x)|
\end{align}
for $f\in C^{\infty}(\phi_{i}(V_{i})).$ Here, $m$ is the index of a fundamental increasing sequence 
$(K_m^{(i)})$ of compact subsets of $\phi_i(V_i).$ 
For further details, see \cite{dieudonne1972treatise}, Chapter XVII.
Then, one can define the space of \textit{distributional sections} by
\begin{equation}
    D'(X,E) := D(X, E^*\otimes\hat{\Lambda}^n(X))' \label{dist-sec-def} 
\end{equation}
Here, $E^*$ is the dual bundle of $E$ and $\hat{\Lambda}^n(X)$ denotes the bundle of densities on $X$. 
For these standard notions, see e.g. \cite{grosser2013geometric, wagner2010distributions}.
One can embed $D(X, E)$ into $D'(X, E)$ by
\begin{equation}
    D(X, E) \hookrightarrow D'(X, E): \psi \mapsto T_\psi, \ 
\langle T_\psi,f\rangle:=\int_X \text{trace}(\psi\otimes f) \label{dsectionEmbed}
\end{equation}
where $\text{trace}(\psi\otimes f)\in L_{\text{loc}}^1(X,\hat{\Lambda}^{n}(X))$ defines an integrable Radon measure on $X$, 
for any $\psi\in D(X, E) \text{ and }f\in D(X, E^*\otimes\hat{\Lambda}^{n}(X))$.
We now specialize these results for general vector bundles to the cases of interest. 

Let $D((\partial B)_T; \mathcal{T}^*(\partial B)_T)$ denote the space of smooth sections with compact support 
of the cotangent bundle ${\mathcal T}^*(\partial B)_T$. Note that 
the tangent spaces are finite-dimensional at each $(\bx,t)\in (\partial B)_T$ and 
thus $\mathcal{T}^*(\partial B)_T\simeq \mathcal{T}(\partial B)_T$ as a bundle isomorphism. 
For $\psi\in D((\partial B)_T; \mathcal{T}^*(\partial B)_T)$ and $(\mathbf{x},t)\in V_{i}\subset(\partial B)_T$
\begin{align}
    \biota_T(\psi(\mathbf{x},t)) &= (\mathbf{x},t,\bu,v),\; \text{with } \bu\in \mathcal{T}^*_{\mathbf{x}}\partial B\subset\mathbb{R}^3,\; v\in \mathcal{T}_t(0,T)=\mathbb{R}
\end{align}
By Prop.16.36 in \cite{lee2013smooth},  $\hat{\Lambda}^{3}((\partial B)_T)$ is a smooth line bundle of $(\partial B)_T$
and as a consequence of 15.29 in  \cite{lee2013smooth}, this density bundle is trivialized by the Riemannian volume form. Thus, 
we may identify 
\begin{align}
    D((\partial B)_T; \mathcal{T}^*(\partial B)_T) &\longleftrightarrow  D((\partial B)_T, \mathcal{T}^*(\partial B)_T\otimes\hat{\Lambda}^3((\partial B)_T))
    \\
    \chi &\longleftrightarrow  \chi\, dS\,dt
\end{align}
where $dS$ is the volume form of $\partial B$ (surface area). In that case, by the general definition
\eqref{dist-sec-def} applied to the tangent bundle 
\be D'((\partial B)_T, \mathcal{T}(\partial B)_T)=D((\partial B)_T, \mathcal{T}^*(\partial B)_T)', \ee 
and we may embed
\begin{align}\label{embed-T} 
    D((\partial B)_T; \mathcal{T}(\partial B)_T) &\hookrightarrow D'((\partial B)_T, \mathcal{T}(\partial B)_T)\\
    \chi &\mapsto T_\chi, \quad \langle T_\chi,\psi\rangle = 
    \int_{(\partial B)_T} \langle\psi,\chi\rangle \, dS\, dt  
 \end{align}
for all $\psi\in  D((\partial B)_T; \mathcal{T}^*(\partial B)_T)$ and 
$\chi\in  D((\partial B)_T; \mathcal{T}(\partial B)_T).$

Likewise, $D((\partial B)_T; \mathcal{N}^*(\partial B)_T)$ denotes the space of smooth sections with compact support of the conormal bundle ${\mathcal N}^*(\partial B)_T\simeq {\mathcal N}(\partial B)_T$, so that for $\bpsi\in D((\partial B)_T; \mathcal{N}^*(\partial B)_T)$ and $(\mathbf{x},t)\in V_{i}\subset(\partial B)_T$,
\begin{align}
    \biota_N(\bpsi(\mathbf{x},t)) &= (\mathbf{x},t,\bu,0),\; \text{with } \bu\in \mathcal{N}^*_{\mathbf{x}}\partial B=\{\bn(\bx)\mathbb{R}\}.
\end{align}
Similarly as before, we may identify 
\begin{align}
    D((\partial B)_T; \mathcal{N}^*(\partial B)_T) &\longleftrightarrow  D((\partial B)_T, \mathcal{N}^*(\partial B)_T\otimes\hat{\Lambda}^3((\partial B)_T))
    \\
    \bchi &\longleftrightarrow  \bchi\, dS\,dt
\end{align}
In that case, by the general definition
\eqref{dist-sec-def} applied to the normal bundle 
\be D'((\partial B)_T, \mathcal{N}(\partial B)_T)=D((\partial B)_T, \mathcal{N}^*(\partial B)_T)', \ee 
and we may embed
\begin{align}
    D((\partial B)_T; \mathcal{N}(\partial B)_T) &\hookrightarrow D'((\partial B)_T, \mathcal{N}(\partial B)_T)\\
    \bchi &\mapsto T_{\scriptsize \bchi}, \quad \langle T_{\scriptsize \bchi},\bpsi\rangle = 
    \int_{(\partial B)_T} \langle\bpsi,\bchi\rangle \, dS\, dt  \label{embed-N} 
\end{align}
for all $\bpsi\in  D((\partial B)_T; \mathcal{N}^*(\partial B)_T)$ and 
$\bchi\in  D((\partial B)_T; \mathcal{N}(\partial B)_T).$

\subsection{Extensions}
The notion of an \textit{extension operator} allows us to identify functions in the 
interior domain $\Omega\times (0,T)$ with sectional distributions of ${\mathcal T}(\partial B)_T$
and of ${\mathcal N}(\partial B)_T.$ Beginning with the tangent bundle, we define 
$\mathcal{E}_{\mathcal{T}}$ as the set of all linear operators
\begin{align}
    \textbf{Ext}: \bpsi\in D((\partial B)_T,\mathcal{T}^*(\partial B)_T) 
    \mapsto \barphi \in \Bar{D}(\Bar{\Omega}\times (0,T),\mathbb{R}^3\times\mathbb{R}) 
\end{align} 
satisfying pointwise equality (\ref{T-extend}) and continuous in the sense that 
for all multi-indices $\alpha = (\alpha_1, \alpha_2, \alpha_3,\alpha_4)$ with $|\alpha| \le N$, 
$\forall (\bx,t)\in\Bar{\Omega}\times(0,T)$ and $\forall m>0$
\begin{align}
    |D^{\alpha}\barphi(\mathbf{x},t)|
    =|D^{\alpha}\textbf{Ext}(\bpsi)(\mathbf{x},t)|\lesssim \sup_{i\in I}p_{N,m,i}(\bpsi)\label{extensionBound}
\end{align}
where, $\lesssim$ denotes inequality with constant prefactor depending on the domain $(\partial B)_T$ 
and the extension operator $\textbf{Ext}$. Note that for $(\bx,t)\in\partial\Omega\times(0,T)$, the 
derivatives $D^{\alpha}$ with non-vanishing spatial indices $\alpha_i,i=1,2,3$ are one-sided derivatives,
which according to definition \eqref{Dbar} may be calculated as $D^\alpha\varphi = 
D^\alpha\phi|_{\Bar{\Omega}\times(0,T)}$ for $\phi\in C_c^{\infty}(\mathbb{R}^3\times(0,T),{\mathbb R}^3),$ 
independent of the choice of $\phi.$ Furthermore, if $\bpsi\in D((\partial B)_T,\mathcal{T}^*(\partial B)_T)$
is a {\it space-like section} of $\mathcal T^*(\partial B)_T$, so that \\
\be 
\biota_T(\bpsi(\mathbf{x},t)) = (\mathbf{x},t,\bu,0),\; \text{with } \bu\in \mathcal{T}^*_{\mathbf{x}}\partial B\subset\mathbb{R}^3
\ee
for all $(\bx,t)\in (\partial B)_T,$ then we may require that 
\begin{align}
    \textbf{Ext}: \bpsi\in D((\partial B)_T,\mathcal{T}^*(\partial B)_T) 
    \mapsto \barphi \in \Bar{D}(\Bar{\Omega}\times (0,T),\mathbb{R}^3\times\{0\}) 
    \lb{sp-cond} 
\end{align} 

We show that the set ${\mathcal E}_T$ is non-empty, by constructing such an extension operator explicitly.
We define $\ext^0_T\in {\mathcal E}_T$ as a map 
\begin{align}
    \ext^0_T: \bpsi\in D((\partial B)_T,\mathcal{T}^*(\partial B)_T) 
    \mapsto \barphi \in \Bar{D}(\Bar{\Omega}\times (0,T),\mathbb{R}^3\times{\mathbb R}) 
\end{align}
by the explicit formula
\begin{align}
    \barphi(\mathbf{x},t) = 
        \begin{cases} \exp\left(-\frac{d(\mathbf{x})}{\epsilon-d(\mathbf{x})}\right)
            (\textbf{Proj}_{st}\circ\biota_T\circ\bpsi)(\pi(\mathbf{x}),t), 
            & d(\mathbf{x})<\epsilon\\
            0 & d(\mathbf{x})\ge \epsilon    
        \end{cases}\label{ext0}
\end{align}
for any $\epsilon<\eta(\Omega).$ Then $\ext^0_T$ is clearly linear by the linearity of  $\biota_T$
and satisfies $\barphi|_{\partial B} =(\textbf{Proj}_{st}\circ\biota_T)(\bpsi)$. 
$\barphi$ is smooth by the smoothness of distance function 
$d$ and projection $\pi$ in $\Bar{\Omega}_\epsilon.$ One can easily obtain the bound (\ref{extensionBound}) 
for $\ext^0_T$ by product rule and chain rule in calculus. Thus, $\ext^0_T$ is continuous.
In particular, for a space-like section $\bpsi\in D((\partial B)_T,\mathcal{T}^*(\partial B)_T),$
the condition \eqref{sp-cond} holds, so that we may take $\barphi\in \Bar{D}(\Bar{\Omega}\times(0,T),\mathbb{R}^3)$
with 
\begin{align}
     \barphi|_{\partial B} = (\textbf{Proj}_{s}\circ\biota_T)(\bpsi), \;\;\;\;
    \barphi(\mathbf{x},t)\perp\mathbf{n}(\pi(\mathbf{x})) \text{ in } \Omega_{\eta(\Omega)}.
\end{align}
As a consequence of the second property, $\ext^0_T$ satisfies also the natural condition \eqref{extCond1} and $\ext^0_T\in \widetilde{{\mathcal E}}_T\neq \emptyset.$ 

Similarly, we can define a set $\mathcal{E}_{\mathcal{N}}$, consisting of maps 
\be 
    \textbf{Ext}: \bpsi\in D((\partial B)_T,\mathcal{N}^*(\partial B)_T) \mapsto 
    \barphi \in \Bar{D}(\Bar{\Omega}\times (0,T),\mathbb{R}^3), 
\ee
which are linear, continuous and satisfy 
\begin{align}
    \barphi|_{\partial B} = (\textbf{Proj}_{s}\circ\biota_N)(\bpsi).
\end{align}
The set $\mathcal{E}_{\mathcal{N}}$ is non-empty, because $\ext_N^0,$ defined for $\epsilon<\eta(\Omega)$ by 
\begin{align}
    \barphi(\mathbf{x},t) = 
        \begin{cases} \exp\left(-\frac{d(\mathbf{x})}{\epsilon-d(\mathbf{x})}\right)
            (\mathbf{Proj}_s\circ\biota_N\circ\bpsi)(\pi(\mathbf{x}),t), 
            & d(\mathbf{x})<\epsilon\\
            0 & d(\mathbf{x})\ge \epsilon  
        \end{cases}\label{extN0}
\end{align}
for any $\bpsi\in D((\partial B)_T,\mathcal{N}^*(\partial B)_T),$ provides an explicit 
example which satisfies also the condition
\be 
    \barphi(\mathbf{x},t)\parallel\mathbf{n}(\pi(\mathbf{x})) \text{ in } \Omega_{\eta(\Omega)} .
\ee 

One can use extension operators to identify ${\bf F}\in D'((0,T),C^\infty(\Bar{\Omega},{\mathbb R}^3))$
with sectional distributions of the tangent and normal bundles. For example, for some $\ext\in \cE_{\cT}$
we define
\begin{equation}
    \begin{aligned}
        \ext^*: D'((0,T),C^\infty(\Bar{\Omega},{\mathbb R}^3)) &\to D'((\partial B)_T, \cT(\partial B)_T)\\
        {\mathbf F} &\mapsto \ext^*({\mathbf F})
    \end{aligned}
\end{equation}
as follows:
\begin{align}
    \langle\ext^*({\mathbf F}),\bpsi\rangle := \langle {\mathbf F},\ext(\bpsi)\rangle
\end{align}
for all $\bpsi\in D((\partial B)_T, \cT^*(\partial B)_T)$.
Linearity and continuity properties of $\ext^*({\mathbf F})$ follow from those of $\ext\in\cE_{\cT^*}$,
so that $\ext^*({\mathbf F})\in D'((\partial B)_T, \cT(\partial B)_T).$ Note that this identification 
depends on the choice of the extension operator $\ext$. Likewise, we can define 
\be 
    \ext^*: D'((0,T),C^\infty(\Bar{\Omega},{\mathbb R}^3))\to D'((\partial B)_T, \cN(\partial B)_T)
\ee
for each $\ext\in\cE_{\cN}$, in exactly the same manner. 

\section{Proof of Theorem 1} 
\noindent
The proof will proceed in steps. First, note that $\btau_w^{\nu} = 
\nu\bomega^{\nu}|_{(\partial B)_T}\btimes\bn$ can be embedded into $D((\partial B)_T, \mathcal{T}(\partial B)_T)$ by
\begin{align}
    \btau_w^{\nu} \mapsto \big((\mathbf{x},t)
    \mapsto(\mathbf{x},t,\btau_w^{\nu}(\mathbf{x},t),0)\big)\label{embeddedStress}
\end{align}
which can be further embedded into $D'((\partial B)_T, \mathcal{T}(\partial B)_T)$ by (\ref{embed-T}). 
For the rest of this article, we abuse the notation $\btau_w^{\nu}$ to mean both vector fields on $(\partial B)_T$ 
and smooth sections (\ref{embeddedStress}) in $D((\partial B)_T, \mathcal{T}(\partial B)_T)$, according to the context.
We then show that $\langle\btau_w^{\nu},\bpsi\rangle$ for any $\bpsi\in D((\partial B)_T, \mathcal{T}^*(\partial B)_T)$ 
converges to a quantity denoted $\langle\btau_w,\bpsi\rangle.$  Finally, we prove that $\btau_w$ is a continuous 
linear functional on $D((\partial B)_T, \mathcal{T}^*(\partial B)_T),$ thus obtaining the convergence (\ref{tau-lim}) 
in the sense of distributional sections of the tangent bundle. 

Similarly, wall pressure stress $p_w^\nu\bn$ is embedded into $D((\partial B)_T, \mathcal{N}(\partial B)_T)$ by
\begin{align}
    p_w^{\nu}\bn \mapsto \big((\mathbf{x},t)
    \mapsto(\mathbf{x},t,p_w^{\nu}(\mathbf{x},t)\bn(\bx),0)\big)\label{embeddedPressure}
\end{align}
which can be further embedded into $D'((\partial B)_T, \mathcal{N}(\partial B)_T)$ by (\ref{embed-N}).
An analogous argument shows that $\langle p_w^\nu\bn,\bpsi\rangle \to\langle p_w\bn,\bpsi\rangle$ for 
all $\bpsi\in D((\partial B)_T, \mathcal{N}^*(\partial B)_T),$ with a suitable element 
$p_w\bn\in D'((\partial B)_T, \mathcal{N}(\partial B)_T).$

\subsection{Convergence of skin friction $\btau_w^{\nu}$ to $\btau_w$}\lb{sec3.1} 
Consider an arbitrary extension operator $\ext\in\cE_{\cT}$, 
and a smooth section $\bpsi\in D((\partial B)_T, \mathcal{T}^*(\partial B)_T)$.
Let $\barphi = \ext(\bpsi)$ so that $\barphi\in\Bar{D}(\Bar{\Omega}\times(0,T),\mathbb{R}^3)$ 
and $\barphi\bdot\bn = 0$ on $(\partial B)_T$.
Integrating the Navier-Stokes equations (\ref{NS1}) against $\barphi$ yields
\begin{equation}
    \begin{aligned}\label{basicT} 
        & -\int_0^T\int_{\Omega}\partial_t\barphi\bdot \mathbf{u}^{\nu}+\grad\barphi\bdots[\mathbf{u}^{\nu}\otimes\mathbf{u}^{\nu} 
        + p^{\nu}\mathbf{I}]\,dV dt\\
        &-\int_0^T \int_{\Omega} \nu \triangle \barphi \bdot\bu^{\nu}\,dV dt 
        =-\int_0^T\int_{\partial\Omega} \nu\frac{\partial \bu^\nu}{\partial n} \bdot\barphi|_{\partial\Omega} \,dS\,dt. 
    \end{aligned}
\end{equation}
As a useful shorthand, we write this as 
\begin{align}
    -\dbangle{\mathbf{u}^{\nu}}{\partial_t\barphi} - \dangle{\mathbf{u}^{\nu}\otimes\mathbf{u}^{\nu}\bdots\grad\barphi}
    -\dbangle{p^{\nu}}{\grad\bdot\barphi} - \dbangle{\nu\bu^{\nu}}{\triangle\barphi} = 
    - \langle \btau_w^{\nu},\bpsi\rangle \label{local_momentum}
\end{align}
where $\dbangle{}{}$ denotes the integration over space-time domain $\Omega\times(0,T)$ and 
\begin{align}
    \langle \btau_w^{\nu},\bpsi\rangle = \int_0^T\int_{\partial\Omega}
    \langle\bpsi,\btau_w^{\nu}\rangle \, dS\,dt  
    = \int_0^T\int_{\partial\Omega} \nu\frac{\partial \bu}{\partial n}\bdot\barphi|_{\partial\Omega}\, dS\,dt. 
\end{align}

By Cauchy-Schwartz, 
\begin{align}
    \left|\dbangle{\nu\bu^{\nu}}{\triangle\barphi}\right|
    &\le \nu\sqrt{\int_0^T\int_{\Omega} |\triangle\barphi|^2\, dV dt}
    \sqrt{\int_0^T\int_{\text{supp}(\barphi)}
    |\bu^\nu|^2 \,dV dt}\\
    &\to0,\text{ as }\nu\to 0,
\end{align}
as a consequence of the assumptions \eqref{L2Conv},\eqref{uBBound} on velocity $\bu^\nu.$

The convergence of the rest of the lefthand side of (\ref{local_momentum}) as $\nu\to0$ follows from 
the following elementary lemma: 
\begin{lemma}\label{lemma1}
    If $f^{\nu}$ converges weakly to $f$ in $L^p((0,T),L_{\text{loc}}^p(\Omega))$, $1\leq p<\infty,$
    and if $f^{\nu}$ is uniformly bounded in $L^p((0,T),L^p(\Omega_{\epsilon}))$ for
    sufficiently small $\epsilon>0,$  
    then $f\in L^p((0,T),L^p(\Omega_{\epsilon}))$, 
    and for $\varphi\in \Bar{D}(\Bar{\Omega}\times(0,T))$, we have the following limit
    \begin{align}
        \lim_{\nu\to0}\int_0^T\int_{\Omega}\varphi f^{\nu}\, dV dt 
        = \int_0^T\int_{\Omega}\varphi f \, dV dt\label{lemma1limit}
    \end{align}
\end{lemma}

\begin{proof}
    Let $M_\epsilon=\sup_{\nu>0}\norm{f^{\nu}}_{L^p((0,T),L^p(\Omega_{\epsilon}))} < \infty$. 
    Then let $\epsilon_n = 2^{-n}\epsilon$ for $n\ge0$ and $\Gamma_n = \Omega_{\epsilon_n}\backslash\Omega_{\epsilon_{n+1}}$. 
    Then $\Omega_{\epsilon} =\cup_{n=0}^{\infty}\Gamma_n$ and the union is a disjoint union.
    With weak lower-semicontinuity of the $L^p$-norm and Fatou's lemma, we have 
    \begin{align}
        &\int_{(0,T)\times \Omega_{\epsilon}}|f|^p\, dV\,dt
        = \sum_{n=0}^{\infty}\int_{(0,T)\times \Gamma_n}|f|^p\,dV\,dt
        \le\sum_{n=0}^{\infty}\liminf_{\nu\to0}\int_{(0,T)\times \Gamma_n}|f^{\nu}|^p\,dV\,dt\\
        &\le\liminf_{\nu\to0}\sum_{n=0}^{\infty}\int_{(0,T)\times \Gamma_n}|f^{\nu}|^p\, dV\,dt
        =\liminf_{\nu\to0}\int_{(0,T)\times \Omega_{\epsilon}}|f^{\nu}|^p\, dV\,dt\le M^p_\epsilon < \infty
    \end{align}
Thus, we obtain that $f\in L^p((0,T),L^p(\Omega_{\epsilon}))$. 

Furthermore, for any $0<\delta<\epsilon$ we obtain by H\"older inequality and the 
uniform $L^p$ bound on $f^{\nu}$ that for $\frac{1}{p}+\frac{1}{p'}=1$
    \begin{align}
        \sup_{\nu>0}\left|\int_{(0,T)\times \Omega_{\delta}}\varphi f^{\nu} dV\,dt \right|
        \le \norm{\varphi}_{L^{p'}((0,T)\times \Omega_{\delta})} M^p_\epsilon 
    \end{align}
with an identical bound for the limit function $f.$ As a consequence
\begin{eqnarray} 
&& \left|\int_{(0,T)\times \Omega}\varphi f^{\nu} dV\,dt -
\int_{(0,T)\times \Omega}\varphi f dV\,dt \right|
\leq 2  \norm{\varphi}_{L^{p'}((0,T)\times \Omega_{\delta})} M^p_\epsilon \cr 
&& \hspace{50pt} + 
\left|\int_{(0,T)\times \Omega^\delta}\varphi f^{\nu} dV\,dt -
\int_{(0,T)\times \Omega^\delta}\varphi f dV\,dt \right| 
\end{eqnarray}  
where $\Omega^\delta:=\Omega\backslash\Omega_\delta.$
Using convergence of $f^{\nu}$ to $f$ weakly in $L^p((0,T),L_{\text{loc}}^p(\Omega))$
and $\lim_{\delta\to 0} \norm{\varphi}_{L^{p'}((0,T)\times\Omega_{\delta})}=0,$ we conclude. 
\hfill $\qquad \qquad $
\end{proof}

Conditions \eqref{L2Conv},\eqref{pBulkBound} imply that, at least along a subsequence, 
$\bu^{\nu},$ $\bu^\nu\otimes\bu^\nu,$ $p^{\nu}$ have local weak convergence to $\bu,$ $\bu\otimes\bu,$ $p$ respectively. 
Then by Lemma \ref{lemma1}, 
\begin{align}
    \mathbf{u} \in L^2((0,T),L^{2}(\Omega_{\epsilon})),\;\;\;\; 
    p\in L^1((0,T), L^{1}(\Omega_{\epsilon}))\label{nearBoundary}
\end{align}
and as $\nu\to0$, the left hand side of (\ref{local_momentum}) converges to
\begin{align}
    -\dbangle{\mathbf{u}}{\partial_t\barphi} - \dangle{\mathbf{u}\otimes\mathbf{u}\bdots\grad\barphi}
    -\dbangle{p}{\grad\bdot\barphi}
:=
\langle\btau_w,\bpsi\rangle. \label{taulim}
\end{align}
As $\bpsi$ was arbitrary, , we conclude that
\begin{align}
    \lim_{\nu\to0}\langle\btau_w^{\nu},\bpsi\rangle = 
    \langle\btau_w,\bpsi\rangle, \;\forall\bpsi\in D((\partial B)_T,\mathcal{T}^*(\partial B)_T)
\end{align}

\subsection{$\btau_w$ is a distributional section}\lb{sec3.2} 
Linearity of $\btau_w$ follows easily from the linearity of $\ext$ and the definition (\ref{taulim}).
Then, it suffices to prove the continuity of $\btau_w$. 
Let $K$ be a compact subset of $(\partial B)_T$. Then there exists a finite set $J$ such that 
$K\subset\cup_{i\in J}V_i$, where $\cup(V_i, \phi_i)$ is a smooth structure of $(\partial B)_T$.
Furthermore, there exists some $m_0>0$ such that for each $i\in J,$ 
$\phi_i(K\cap V_i)\subset K_{m_0}^{(i)}$ for a compact set $K_{m_0}^{(i)}$ in the fundamental 
sequence of $\phi_i(V_i).$ Therefore, for all $\bpsi\in D((\partial B)_T, \mathcal{T}^*(\partial B)_T)$ 
supported on $K$ and for all $m\geq m_0$
\begin{align}
    \dbangle{\mathbf{u}}{\partial_t\barphi} \lesssim \norm{\bu}_{L^2(\text{supp}(\barphi))}\sup_{i\in I}p_{1,m,i}(\bpsi)\\
    \dangle{\mathbf{u}\otimes\mathbf{u}\bdots\grad\barphi}\lesssim \norm{\bu}^2_{L^2(\text{supp}(\barphi))}\sup_{i\in I}p_{1,m,i}(\bpsi)
\end{align}
where $\barphi = \ext(\bpsi),$ so that $\text{supp}(\barphi)$ is a compact subset of $\Bar{\Omega}\times(0,T)$ 
by definition \eqref{Dbar}. 
Here, $\lesssim$ denotes inequality up to a constant prefactor, depending on $K,$ $\ext$.
Note that $\bu$ is bounded in $L^2(\text{supp}(\barphi))$, 
as a result of interior boundedness (\ref{L2Conv}) and near-boundary boundedness 
(\ref{nearBoundary}) in $L^2$. 
Similarly, for all $\bpsi\in D((\partial B)_T, \mathcal{T}^*(\partial B)_T)$ 
supported on $K$ and for all $m\geq m_0,$ $p$ is bounded in $L^1(\text{supp}(\barphi))$ and 
\begin{align}
    \dbangle{p}{\grad\bdot\barphi}\lesssim \norm{p}_{L^1(\text{supp}(\barphi))}\sup_{i\in I}p_{1,m,i}(\bpsi). 
\end{align}
In conclusion, $\btau_w$ is continuous and $\btau_w$ is thus a well defined distribution 
in $D'((\partial B)_T),\mathcal{T}(\partial B)_T)$ for each $\ext\in {\mathcal E}_{{\mathcal T}}.$

Note that $\btau_w$ is independent of $\textbf{Ext}\in\mathcal{E}_{\mathcal{T}}.$
Indeed, by combining the result in this section with that in \ref{sec3.1}, we see 
for each $\ext\in {\mathcal E}_{{\mathcal T}}$ that $\lim_{\nu\to 0}\btau_w^{\nu}=\btau_w$
in the standard topology of $D'((\partial B)_T),\mathcal{T}(\partial B)_T).$ Since 
such limits are unique, $\btau_{w}$ is independent of the choice of extension operator
and depends only upon the subsequence $\nu_k\to 0$ chosen to obtain the limiting weak Euler 
solution $(\bu,p).$

\subsection{Pressure stress $p_w\bn$}
Consider now instead an arbitrary extension operator $\ext\in\cE_{\cN}$
and a smooth section $\bpsi\in D((\partial B)_T, \mathcal{N}^*(\partial B)_T)$.
Let $\barphi = \ext(\bpsi)$ so that $\barphi\in\Bar{D}(\Bar{\Omega}\times(0,T),\mathbb{R}^3)$ 
and $\barphi\parallel\bn$ on $(\partial B)_T$.
Integrating the Navier-Stokes equations (\ref{NS1}) against $\barphi$ yields
\begin{equation}
    \begin{aligned}\label{basicN} 
        & -\int_0^T\int_{\Omega} \left[\,\partial_t\barphi\bdot \mathbf{u}^{\nu}+\grad\barphi\bdots(\mathbf{u}^{\nu}\otimes\mathbf{u}^{\nu} 
        + p^{\nu}\mathbf{I})\,\right]\,dV\,dt\\
        &+\nu\int_0^T \int_{\Omega}(\triangle\barphi)\bdot\bu^{\nu}\,dV\,dt 
        =\int_0^T\int_{\partial\Omega}p^{\nu}|_{(\partial B)_T}\mathbf{n}\bdot\barphi|_{(\partial B)_T} dS\,dt
    \end{aligned}
\end{equation}
On the other hand,
\begin{align}
    \langle p_w^{\nu}\bn,\bpsi\rangle 
    = \int_0^T\int_{\partial B}\langle \bpsi,p^{\nu}_w\bn\rangle \, dS\,dt 
    = \int_0^T\int_{\partial B}p^{\nu}|_{(\partial B)_T}\mathbf{n}\bdot\barphi|_{(\partial B)_T} dS\,dt.
\end{align}
In shorthand, 
\begin{align}
    -\dbangle{\mathbf{u}^{\nu}}{\partial_t\barphi} - \dangle{\mathbf{u}^{\nu}\otimes\mathbf{u}^{\nu}\bdots\grad\barphi}
    -\dbangle{p^{\nu}}{\grad\bdot\barphi} + \dbangle{\nu\bu^{\nu}}{\triangle\barphi} = \langle p_w^{\nu}\bn,\bpsi\rangle \label{local_momentum_n}
\end{align}
By an analogous argument as that used to prove convergence of $\btau_w^{\nu}$ to $\btau_w$, 
it follows that (\ref{local_momentum_n}) in the limit $\nu\to 0$ yields 
for all $\bpsi\in D((\partial B)_T, \mathcal{N}^*(\partial B)_T)$ 
\begin{align}
    \langle p_w^\nu\bn,\bpsi\rangle\xrightarrow{\nu\to 0} 
    \langle p_w\bn,\bpsi\rangle := -\dbangle{\mathbf{u}}{\partial_t\barphi} - \dangle{\mathbf{u}\otimes\mathbf{u}\bdots\grad\barphi}
    -\dbangle{p}{\grad\bdot\barphi}\label{pressure_lim}
\end{align}
and $p_w\bn\in D'((\partial B)_T),\mathcal{N}(\partial B)_T)$, independent of the extension
$\ext\in{\mathcal E}_{\cN}.$

\section{Proof of Theorem 2} 
We give here a detailed proof only of the result \eqref{tangentialLimit} on the convergence 
in the space of distributional sections of the tangent bundle. The statement 
\eqref{normalLimit} on convergence in the space of distributional sections of the normal bundle 
is proved by a very similar argument, which is left to the reader. 

\subsection{Proof of a lemma} We first prove: 
\begin{lemma}\label{lemma2}
    Let $K$ be a compact subset of $\bar{\Omega}$. 
    Then for any $f\in L^p((0,T),L^p_{\text{loc}}(\Omega))$ $\cap L^p((0,T),L^p(\Omega_{\epsilon}))$, 
    with $1\le p<\infty$ and $\epsilon>0$, we have for $0<\ell<h$,
    \begin{align}
        \eta_{h,\ell}\Bar{f}_{\ell}\xrightarrow[L^p((0,T),L^p(K))]{h,\ell\to0}f\label{lemma2conv}
    \end{align}
\end{lemma}
\begin{proof}
    Let $\delta>0$ be a sufficiently small number with $\ell<h<\delta<\frac{\epsilon}{2}$.
    It is well-known for $f\in L^p_{\text{loc}}(\Omega)$ that $\Bar{f}_{\ell}\to f$ in $L^p_{\text{loc}}(\Omega)$ 
    (e.g. see Appendix C.5 of \cite{evans2010partial}). 
    In particular, for a.e. $t\in(0,T)$,
    \begin{align}
        \lim_{h,\ell\to0}\norm{\eta_{h,\ell}\Bar{f}_{\ell}-f}_{L^p(\Omega^{\delta}\cap K)} = 0
        \label{cout} 
    \end{align}
    and 
    \be \norm{\eta_{h,\ell}\Bar{f}_{\ell}-f}_{L^p(\Omega^{\delta}\cap K)}\leq 2\norm{f}_{L^p(K_{\epsilon}\cap\Omega)} \ee
    for $K_\epsilon=\{\bx\in {\mathbb R}^3: \exists \by\in K,\, |\bx-\by|<\epsilon\}.$
    On the other hand, for a.e. $t\in(0,T)$,
    \begin{align}
        \norm{\eta_{h,\ell}\Bar{f}_{\ell}-f}_{L^p(\Omega_{\delta}\cap K)} 
        &\leq \norm{\eta_{h,\ell}\Bar{f}_{\ell}-f}_{L^p(\Omega_{\delta})} \\
        &= \norm{\eta_{h,\ell}\Bar{f}_{\ell}-f}_{L^p(\Omega_{\delta}\backslash\Omega_h)}
        + \norm{f}_{L^p(\Omega_{h})}\\
        &\le2\norm{f}_{L^p(\Omega_{2\delta})}
        +\norm{f}_{L^p(\Omega_{h})}\\
        &\le3\norm{f}_{L^p(\Omega_{2\delta})}
        \ \le \ 3\norm{f}_{L^p(\Omega_{\epsilon})}
        \label{cin} 
    \end{align}
    Then by combining \eqref{cout},\eqref{cin} 
    \begin{align}
        \limsup_{h,\ell\to0}\norm{\eta_{h,\ell}\Bar{f}_{\ell}-f}_{L^p(K)}
        \le3\norm{f}_{L^p(\Omega_{2\delta})}\xrightarrow{\delta\to0}0
        \label{lemma2conv1}
    \end{align}
    where the latter follows by dominated convergence theorem. 
    Since (\ref{lemma2conv1}) is true for a.e. $t\in(0,T),$
    one obtains (\ref{lemma2conv}), or convergence in $L^p((0,T),L^p(K)).$
\end{proof}

\subsection{Proof of Theorem 2} 
Take any extension operator $\textbf{Ext}\in\mathcal{E}_{\mathcal{T}}$
and smooth section $\bpsi\in D((\partial B)_T, \mathcal{T}^*(\partial B)_T)$.
Let $\barphi = \ext(\bpsi)$ so that $\barphi\in\Bar{D}(\Bar{\Omega}\times(0,T),\mathbb{R}^3)$ 
and $\barphi\bdot\bn = 0$ on $(\partial B)_T$.
Integrating the coarse-grained Euler equations (\ref{coarseGrainEuler}) against $\barphi$ yields
\begin{align}
    &\int_0^T\int_{\Omega}\partial_t\barphi\bdot(\eta_{h,\ell}\Bar{\mathbf{u}}_{\ell})\,dV\,dt
    +\int_0^T\int_{\Omega}\grad\barphi\bdots\eta_{h,\ell}(\Bar{\mathbf{T}}_{\ell}+\Bar{p}_{\ell}\mathbf{I})
    \,dV\,dt\label{cgEuler1}\\
    &=-\int_0^T\int_{\Omega}\barphi\bdot(\grad\eta_{h,\ell}\bdot\Bar{\mathbf{T}}_{\ell} 
    + \Bar{p}_{\ell}\grad\eta_{h,\ell})\,dV\,dt\label{cgEuler2}
\end{align}
As $\barphi\in\Bar{D}(\Bar{\Omega}\times(0,T),\mathbb{R}^3)$, there exists a compact set 
$K\subset\Bar{\Omega}$ such that 
\begin{align}
    \text{supp}(\barphi)\subset K\times(0,T)\subset\Bar{\Omega}\times(0,T)
\end{align}
By Lemma \ref{lemma2}, as $h,\ell\to0$,
\begin{align}
    \eta_{h,\ell}\Bar{\mathbf{u}}_{\ell}&\to\mathbf{u}\text{ in } L^2((0,T),L^2(K))\\
    \eta_{h,\ell}\Bar{\mathbf{T}}_{\ell}&\to\mathbf{T}\text{ in } L^1((0,T),L^1(K))\\
    \eta_{h,\ell}\Bar{p}_{\ell}&\to p\text{ in } L^1((0,T),L^1(K))
\end{align}
Then, as $h,\ell\to0$, (\ref{cgEuler1}) converges to
\begin{align}
    \int_0^T\int_{\Omega}\partial_t\barphi\bdot\mathbf{u}\,dV\,dt
    +\int_0^T\int_{\Omega}\grad\barphi\bdots[\mathbf{T} + p\mathbf{I}]\,dV\,dt
\end{align}
Thus we obtain from (\ref{cgEuler1})-(\ref{cgEuler2}) that
\begin{align}
    &-\lim_{h,\ell\to0}\int_0^T\int_{\Omega}\barphi\bdot(\grad\eta_{h,\ell}\bdot\Bar{\mathbf{T}}_{\ell} 
    + \Bar{p}_{\ell}\grad\eta_{h,\ell}) \,dV\,dt\\
    &= \int_0^T\int_{\Omega}\partial_t\barphi\bdot\mathbf{u}\, dV\,dt
    +\int_0^T\int_{\Omega}\grad\barphi\bdots[\mathbf{T} + p\mathbf{I}]\,dV\,dt\label{viscStress2}
\end{align}
Then by comparison with $\langle\btau_w,\bpsi\rangle$ defined by (\ref{taulim}), we obtain that
\begin{align}
    -\lim_{h,\ell\to0}\int_0^T\int_{\Omega}\barphi\bdot(\grad\eta_{h,\ell}\bdot\Bar{\mathbf{T}}_{\ell} 
    + \Bar{p}_{\ell}\grad\eta_{h,\ell})\,dV\,dt = \langle \btau_w,\bpsi\rangle
\end{align}
In other words, for any $\ext\in\cE_{\cT}$
\begin{align}
    -\lim_{h,\ell\to0}\ext^*(\grad\eta_{h,\ell}\bdot\Bar{\mathbf{T}}_{\ell} + \Bar{p}_{\ell}\grad\eta_{h,\ell}) 
    = \btau_w\text{ in }D'((\partial B)_T, \mathcal{T}(\partial B)_T)
\end{align}

\subsection{Proof of Corollary 1} 
For any $\ext\in\tilde{\cE}_{\cT}$ as in (\ref{extCond1}),
$\forall\bpsi\in D((\partial B)_T, \mathcal{T}^*(\partial B)_T)$, $\barphi = \textbf{Ext}(\bpsi)$ 
\begin{align}
    \langle\ext^*(\Bar{p}_{\ell}\grad\eta_{h,\ell}),\bpsi\rangle 
    &= \int_0^T\int_{\Omega}\barphi\bdot(\Bar{p}_{\ell}\grad\eta_{h,\ell})\,dV\,dt\\
    &=\int_0^T\int_{\Omega_{h+\ell}\backslash\Omega_{h}}\theta'_{h,\ell}(d(\mathbf{x}))\Bar{p}_{\ell}(\mathbf{x},t)\barphi(\mathbf{x},t)
    \bdot\mathbf{n}(\pi(\mathbf{x}))\,dV\,dt \end{align}
\begin{align} 
\left|  \langle\ext^*(\Bar{p}_{\ell}\grad\eta_{h,\ell}),\bpsi\rangle \right| 
    &\le\norm{\theta_{h,\ell}'\barphi\bdot\mathbf{n}}_{L^{\infty}((\Omega_{h+\ell}\backslash\Omega_h)\times(0,T))}
    \left(\int_0^T\int_{\Omega_{h+\ell}\backslash\Omega_h}|\Bar{p}_{\ell}|\,dV\,dt\right)\\
    &\le\frac{C}{\ell}\norm{\barphi\bdot\mathbf{n}}_{L^{\infty}((\Omega_{h+\ell}\backslash\Omega_h)\times(0,T))}
    \norm{p}_{L^1((\Omega_{h+2\ell}\backslash\Omega_{h-\ell})\times(0,T))}\\
    & \le C' \norm{p}_{L^1((\Omega_{3h})\times(0,T))}\xrightarrow{h,\ell\to0}0
\end{align}
by \eqref{extCond1} and dominated convergence. By comparison with \eqref{tangentialLimit} we obtain that
\begin{align}
    -\lim_{h,\ell\to0}\ext^*(\grad\eta_{h,\ell}\bdot\Bar{\mathbf{T}}_{\ell}) = \btau_w \text{ in } D'((\partial B)_T, \mathcal{T}(\partial B)_T)
\end{align}

\section{Proof of Theorem 3} 
\noindent
Theorem \ref{theorem3} follows from Proposition \ref{proposition2}, 
in conjunction with Theorem \ref{theorem1} \& \ref{theorem2} and Corollary \ref{corollary1}.
We thus prove Proposition \ref{proposition2} in this section. 
We follow the idea in \cite{drivas2018nguyen} by bounding the following term directly 
\begin{align}
    \grad\eta_{h,\ell}\bdot\Bar{\mathbf{T}}_{\ell}(\bx, t) 
    = \theta_{h,\ell}'(d(\mathbf{x}))\mathbf{n}(\pi(\mathbf{x}))\bdot\overline{\mathbf{u}\otimes\mathbf{u}}_{\ell}(\mathbf{x},t)
\end{align}
which is supported in $\Omega_{h+\ell}\backslash\Omega_h\subset\Omega_{3h}\subset\Omega_{\epsilon}$. 
We write, $\forall \mathbf{x}\in\Omega_{h+\ell}\backslash\Omega_h$, a.e. $t\in (0,T)$,
\begin{eqnarray}
    \mathbf{n}(\pi(\mathbf{x}))\bdot\overline{\mathbf{u}\otimes\mathbf{u}}_{\ell}(\mathbf{x},t) 
    &=& \int_{\mathbb{R}^3}G_\ell(\mathbf{r})[\mathbf{n}(\pi(\mathbf{x}))-\mathbf{n}(\pi(\mathbf{x} + \mathbf{r}))]
    \bdot\mathbf{u}\otimes\mathbf{u}(\mathbf{x} + \mathbf{r},t)\,V(d\mathbf{r})\cr
    &&+\int_{\mathbb{R}^3}G_\ell(\mathbf{r})\mathbf{n}(\pi(\mathbf{x} + \mathbf{r}))
    \bdot\mathbf{u}\otimes\mathbf{u}(\mathbf{x} + \mathbf{r},t)\, V(d\mathbf{r})
\end{eqnarray}
Since $\mathbf{n}\circ\pi$ is smooth in $\overline{\Omega_{\epsilon}}$, $\forall\delta>0$, $\exists\rho=\rho(\delta)>0$ s.t.
\begin{align}
    |\mathbf{n}(\pi(\mathbf{x}))-\mathbf{n}(\pi(\mathbf{x} + \mathbf{r}))|\le\delta
\end{align}
for all $\mathbf{x}\in\Omega_{h+\ell}\backslash\Omega_{h}$ and $|r|<\ell<\rho$. Then it follows that
\begin{align}
    |\mathbf{n}(\pi(\mathbf{x}))
    \bdot\overline{\mathbf{u}\otimes\mathbf{u}}_{\ell}(\mathbf{x},t)|
    &\le\left(\delta\norm{\mathbf{u}(t)}_{L^{\infty}(\Omega_{\epsilon})} 
    + \norm{\mathbf{n}\bdot\mathbf{u}(t)}_{L^{\infty}(\Omega_{\varepsilon})}\right)\norm{\mathbf{u}(t)}_{L^{\infty}(\Omega_{\epsilon})}
\end{align}
Using these bounds above, together with the fact that $\norm{\theta_{h,\ell}'(d(\mathbf{x}))}_{L^{\infty}}\le\frac{C}{\ell}$ 
and $|\Omega_{h+\ell}\backslash\Omega_h|\le C'\ell$, we obtain that for $\bpsi\in D((\partial B)_T, \mathcal{T}^*(\partial B)_T)$, $\ext\in\cE_{\cT}$
\begin{align}
    \langle \ext^*(\grad\eta_{h,\ell}\bdot\Bar{\mathbf{T}}_{\ell}),\bpsi\rangle 
    &= \int_0^T\int_{\Omega}\barphi\bdot(\grad\eta_{h,\ell}\bdot\Bar{\mathbf{T}}_{\ell})\,dV\,dt
\end{align}
\begin{align}
    \left|\langle
    \ext^*(\grad\eta_{h,\ell})\bdot\Bar{\mathbf{T}}_{\ell},\bpsi\rangle \right|
    &\le\norm{\barphi}_{L^{\infty}((0,T)\times\Omega)} 
    \int_0^T\int_{\Omega_{h+\ell}\backslash\Omega_h}|\grad\eta_{h,\ell}\bdot\Bar{\mathbf{T}}_{\ell}|\,dV\,dt\\
    &\lesssim \sup_{i\in I}p_{N,m,i}(\bpsi)\times\bigg[\delta\norm{\mathbf{u}}_{L^2((0,T),L^{\infty}(\Omega_{\epsilon}))}^2\\
    &+\norm{\mathbf{n}\bdot\mathbf{u}}_{L^2((0,T),L^{\infty}(\Omega_{\epsilon}))}\norm{\mathbf{u}}_{L^2((0,T),L^{\infty}(\Omega_{\epsilon}))}
    \bigg]
\end{align}
where $\barphi = \textbf{Ext}(\bpsi).$
Thus, by the assumptions on the near wall uniform boundedness of $\mathbf{u}$ (\ref{wallboundedness})
and the continuity of wall normal velocity (\ref{wallnormal}), the first 
result \eqref{prop1a} in Proposition 1 follows:
\begin{align}
    \lim_{h,\ell\to0}\ext^*(\grad\eta_{h,\ell}\bdot\Bar{\mathbf{T}}_{\ell}) = 0 \text{ in } D'((\partial B)_T, \mathcal{T}(\partial B)_T)
\end{align}
It is easy to see that the argument above applies also for all $\ext\in\cE_{\cN}$ and 
$\bpsi\in D((\partial B)_T, \mathcal{N}^*(\partial B)_T)$. Thus, the result \eqref{prop1b} in Proposition 1 also follows: 
\begin{align}
    \lim_{h,\ell\to0}\ext^*(\grad\eta_{h,\ell}\bdot\Bar{\mathbf{T}}_{\ell}) = 0 \text{ in } D'((\partial B)_T, \mathcal{N}(\partial B)_T)
\end{align}

\section*{Acknowledgements} 
We are grateful to Nigel Goldenfeld, Samvit Kumar, Charles Meneveau, Yves Pomeau, Tamer Taki, 
Vlad Vicol and especially Theo Drivas for discussions of this problem. This work was funded 
by the Simons Foundation, via Targeted Grant in MPS-663054 and the Collaboration 
Grant No. MPS-1151713.

\bibliographystyle{plain}
\bibliography{bibliography}

\begin{thebibliography}{10}

\bibitem{achenbach1972experiments}
Elmar Achenbach.
\newblock Experiments on the flow past spheres at very high {R}eynolds numbers.
\newblock {\em Journal of fluid mechanics}, 54(3):565--575, 1972.

\bibitem{bardos2018onsager}
Claude Bardos and Edriss~S Titi.
\newblock Onsager’s conjecture for the incompressible {E}uler equations in
  bounded domains.
\newblock {\em Archive for Rational Mechanics and Analysis}, 228(1):197--207,
  2018.

\bibitem{bardos2019onsager}
Claude Bardos, Edriss~S Titi, and Emil Wiedemann.
\newblock Onsager's conjecture with physical boundaries and an application to
  the vanishing viscosity limit.
\newblock {\em Communications in Mathematical Physics}, 370(1):291--310, 2019.

\bibitem{bardos2013mathematics}
Claude~W Bardos and Edriss~S Titi.
\newblock Mathematics and turbulence: where do we stand?
\newblock {\em Journal of Turbulence}, 14(3):42--76, 2013.

\bibitem{bardos2022c}
Claude~W Bardos and Edriss~S Titi.
\newblock {$C^{0,\alpha}$} boundary regularity for the pressure in weak
  solutions of the $2d$ {E}uler equations.
\newblock {\em Philo. Trans. Roy. Soc. A}, 380(2218):20210073, 2022.

\bibitem{bedrossian2019sufficient}
Jacob Bedrossian, Michele Coti~Zelati, Samuel Punshon-Smith, and Franziska
  Weber.
\newblock A sufficient condition for the {K}olmogorov 4/5 law for stationary
  martingale solutions to the 3{D} {N}avier--{S}tokes equations.
\newblock {\em Communications in Mathematical Physics}, 367(3):1045--1075,
  2019.

\bibitem{buckmaster2019onsager}
Tristan Buckmaster, Camillo De~Lellis, L{\'a}szl{\'o} Sz{\'e}kelyhidi, and Vlad
  Vicol.
\newblock Onsager's conjecture for admissible weak solutions.
\newblock {\em Communications on Pure and Applied Mathematics}, 72(2):229--274,
  2019.

\bibitem{busse2017reynolds}
Angela Busse, Manan Thakkar, and Neil~D Sandham.
\newblock Reynolds-number dependence of the near-wall flow over irregular rough
  surfaces.
\newblock {\em Journal of Fluid Mechanics}, 810:196--224, 2017.

\bibitem{cadot1997energy}
Olivier Cadot, Yves Couder, Adrian Daerr, Stephane Douady, and Arkady Tsinober.
\newblock Energy injection in closed turbulent flows: {S}tirring through
  boundary layers versus inertial stirring.
\newblock {\em Physical Review E}, 56(1):427, 1997.

\bibitem{chen2020kato}
Robin~Ming Chen, Zhilei Liang, and Dehua Wang.
\newblock A {K}ato-type criterion for vanishing viscosity near the {O}nsager's
  critical regularity.
\newblock {\em arXiv preprint arXiv:2007.12746}, 2020.

\bibitem{cheskidov2008energy}
Alexey Cheskidov, Peter Constantin, Susan Friedlander, and Roman Shvydkoy.
\newblock Energy conservation and {O}nsager's conjecture for the {E}uler
  equations.
\newblock {\em Nonlinearity}, 21(6):1233, 2008.

\bibitem{constantin1994onsager}
Peter Constantin, Edriss~S Titi, and {Weinan E}.
\newblock Onsager's conjecture on the energy conservation for solutions of
  {E}uler's equation.
\newblock {\em Communications in Mathematical Physics}, 165(1):207, 1994.

\bibitem{dalembert1749theoria}
Jean le~Rond d{'}Alembert.
\newblock Theoria resistentiae quam patitur corpus in fluido motum, ex
  principiis omnino novis et simplissimis deducta, habita ratione tum
  velocitatis, figurae, et massae corporis moti, tum densitatis compressionis
  partium fluidi.
\newblock manuscript at Berlin-Brandenburgische Akademie der Wissenschaften,
  Akademie-Archiv call number: I-M478, 1749.

\bibitem{dalembert1768paradoxe}
Jean le~Rond d{'}Alembert.
\newblock Paradoxe propos\'e aux g\'eom\`etres sur la r\'esistance des fluides.
\newblock in: Opuscules math\'ematiques, vol. 5 (Paris), Memoir XXXIV, Section
  I, 132--138, 1768.

\bibitem{derosa2023double}
Luigi De~Rosa, Micka{\"e}l Latocca, and Giorgio Stefani.
\newblock On double h{\"o}lder regularity of the hydrodynamic pressure in
  bounded domains.
\newblock {\em Calculus of Variations and Partial Differential Equations},
  62(3):85, 2023.

\bibitem{derosa2024full}
Luigi De~Rosa, Micka{\"e}l Latocca, and Giorgio Stefani.
\newblock Full double h{\"o}lder regularity of the pressure in bounded domains.
\newblock {\em International Mathematics Research Notices}, 2024(3):2511--2560,
  2024.

\bibitem{dieudonne1972treatise}
Jean Dieudonn{\'e}.
\newblock {\em Treatise on Analysis: Volume 3}, volume~3.
\newblock Academic Press, 1972.

\bibitem{drivas2018onsager}
Theodore~D Drivas and Gregory~L Eyink.
\newblock An {O}nsager singularity theorem for turbulent solutions of
  compressible {E}uler equations.
\newblock {\em Communications in Mathematical Physics}, 359(2):733--763, 2018.

\bibitem{drivas2018nguyen}
Theodore~D Drivas and Huy~Q Nguyen.
\newblock Onsager's conjecture and anomalous dissipation on domains with
  boundary.
\newblock {\em SIAM Journal on Mathematical Analysis}, 50(5):4785--4811, 2018.

\bibitem{drivas2019remarks}
Theodore~D Drivas and Huy~Q Nguyen.
\newblock Remarks on the emergence of weak {E}uler solutions in the vanishing
  viscosity limit.
\newblock {\em Journal of Nonlinear Science}, 29(2):709--721, 2019.

\bibitem{duchon2000inertial}
Jean Duchon and Raoul Robert.
\newblock Inertial energy dissipation for weak solutions of incompressible
  {E}uler and {N}avier-{S}tokes equations.
\newblock {\em Nonlinearity}, 13(1):249, 2000.

\bibitem{e2000boundary}
Weinan E.
\newblock Boundary layer theory and the zero-viscosity limit of the
  {N}avier-{S}tokes equation.
\newblock {\em Acta Mathematica Sinica}, 16(2):207--218, 2000.

\bibitem{evans2010partial}
Lawrence~C. Evans.
\newblock {\em Partial Differential Equations}.
\newblock Graduate studies in mathematics. American Mathematical Society, 2010.

\bibitem{eyink2024onsager}
Gregory Eyink.
\newblock Onsager's ‘ideal turbulence’ theory.
\newblock {\em J. Fluid Mech.}, 988:P1, 2024.

\bibitem{eyink1994energy}
Gregory~L Eyink.
\newblock Energy dissipation without viscosity in ideal hydrodynamics {I}.
  {F}ourier analysis and local energy transfer.
\newblock {\em Physica D: Nonlinear Phenomena}, 78(3-4):222--240, 1994.

\bibitem{eyink2008turbulent}
Gregory~L Eyink.
\newblock Turbulent flow in pipes and channels as cross-stream “inverse
  cascades” of vorticity.
\newblock {\em Physics of Fluids}, 20(12):125101, 2008.

\bibitem{eyink2021josephson}
Gregory~L. Eyink.
\newblock {J}osephson-{A}nderson relation and the classical {D}'{A}lembert
  paradox.
\newblock {\em Phys. Rev. X}, 11:031054, Sep 2021.

\bibitem{eyink2022Aonsager}
Gregory~L Eyink, Samvit Kumar, and Hao Quan.
\newblock The {O}nsager theory of wall-bounded turbulence and {T}aylor’s
  momentum anomaly.
\newblock {\em Philosophical Transactions of the Royal Society A},
  380(2218):20210079, 2022.

\bibitem{foote1984regularity}
Robert~L Foote.
\newblock Regularity of the distance function.
\newblock {\em Proceedings of the American Mathematical Society},
  92(1):153--155, 1984.

\bibitem{grosser2013geometric}
Michael Grosser, Michael Kunzinger, Michael Oberguggenberger, and Roland
  Steinbauer.
\newblock {\em Geometric theory of generalized functions with applications to
  general relativity}, volume 537.
\newblock Springer Science \& Business Media, 2013.

\bibitem{guillemin1990geometric}
Victor Guillemin and Shlomo Sternberg.
\newblock {\em Geometric Asymptotics}.
\newblock Mathematical surveys and monographs. American Mathematical Society,
  1990.

\bibitem{hoffman2006simulation}
Johan Hoffman.
\newblock Simulation of turbulent flow past bluff bodies on coarse meshes using
  {G}eneral {G}alerkin methods: drag crisis and turbulent {E}uler solutions.
\newblock {\em Computational Mechanics}, 38(4):390--402, 2006.

\bibitem{isett2018proof}
Philip Isett.
\newblock A proof of {O}nsager's conjecture.
\newblock {\em Annals of Mathematics}, 188(3):871--963, 2018.

\bibitem{jimenez2012cascades}
Javier Jim{\'e}nez.
\newblock Cascades in wall-bounded turbulence.
\newblock {\em Annual Review of Fluid Mechanics}, 44:27--45, 2012.

\bibitem{lee2013smooth}
John~M Lee.
\newblock {\em Introduction to Smooth Manifolds}.
\newblock Springer, 2013.

\bibitem{leobacher2021existence}
Gunther Leobacher and Alexander Steinicke.
\newblock Existence, uniqueness and regularity of the projection onto
  differentiable manifolds.
\newblock {\em Annals of Global Analysis and Geometry}, 60(3):559--587, 2021.

\bibitem{lighthill1963introduction}
Michael~James Lighthill.
\newblock Introduction: {B}oundary layer theory.
\newblock In L.~Rosenhead, editor, {\em Laminar Boundary Theory}, pages
  46--113. Oxford University Press, Oxford, 1963.

\bibitem{morton1984generation}
Bruce~R. Morton.
\newblock The generation and decay of vorticity.
\newblock {\em Geophysical \& Astrophysical Fluid Dynamics}, 28(3-4):277--308,
  1984.

\bibitem{nguyenvanyen2018energy}
Natacha Nguyen~van yen, Matthias Waidmann, Rupert Klein, Marie Farge, and Kai
  Schneider.
\newblock Energy dissipation caused by boundary layer instability at vanishing
  viscosity.
\newblock {\em Journal of Fluid Mechanics}, 849:676--717, 2018.

\bibitem{onsager1949statistical}
Lars Onsager.
\newblock Statistical hydrodynamics.
\newblock {\em Il Nuovo Cimento (1943-1954)}, 6(2):279--287, 1949.

\bibitem{prandtl1905flussigkeitsbewegung}
Ludwig Prandtl.
\newblock {\"U}ber {F}l{\"u}ssigkeitsbewegung bei sehr kleiner {R}eibung.
\newblock In Adolf Krazer, editor, {\em Verhandlungen des dritten
  Internationalen Mathematiker-Kongresses in Heidelberg: vom 8. bis 13. august
  1904}, pages 484--491. B.~G. Teubner, Leipzig, 1905.

\bibitem{quan2024onsager}
Hao Quan and Gregory~L. Eyink.
\newblock {O}nsager theory of turbulence, the {J}osephson-{A}nderson relation,
  and the {D}'{A}lembert paradox.
\newblock Commun. Math. Phys, to appear;
  \url{https://arxiv.org/abs/2206.05326}, 2022.

\bibitem{sohr2012navier}
Hermann Sohr.
\newblock {\em The Navier-Stokes Equations: An Elementary Functional Analytic
  Approach}.
\newblock Modern Birkh{\"a}user Classics. Springer Basel, 2012.

\bibitem{sueur2012kato}
Franck Sueur.
\newblock A {K}ato type theorem for the inviscid limit of the {N}avier-{S}tokes
  equations with a moving rigid body.
\newblock {\em Communications in Mathematical Physics}, 316(3):783--808, 2012.

\bibitem{taylor1915eddy}
Geoffrey~I. Taylor.
\newblock Eddy motion in the atmosphere.
\newblock {\em Philos. Trans. R. Soc. London, Ser. A}, 215(1):1--26, 1915.

\bibitem{tennekes1972first}
Hendrik Tennekes and John~L. Lumley.
\newblock {\em A First Course in Turbulence}.
\newblock MIT Press, 1972.

\bibitem{tomdieck2010algebraic}
Tammo tom Dieck.
\newblock {\em Algebraic Topology: Corrected 2nd Printing, 2010}.
\newblock EMS Textbooks in Mathematics. European Mathematical Society
  Publishing House, 2008.

\bibitem{vasseur2023boundary}
Alexis~F Vasseur and Jincheng Yang.
\newblock Boundary vorticity estimates for navier--stokes and application to
  the inviscid limit.
\newblock {\em SIAM Journal on Mathematical Analysis}, 55(4):3081--3107, 2023.

\bibitem{wagner2010distributions}
Peter Wagner.
\newblock Distributions supported by hypersurfaces.
\newblock {\em Applicable Analysis}, 89(8):1183--1199, 2010.

\bibitem{wiedemann2018weak}
E.~Wiedemann.
\newblock Weak-strong uniqueness in fluid dynamics.
\newblock In C.L. Fefferman, J.C. Robinson, and J.L. Rodrigo, editors, {\em
  Partial differential equations in fluid mechanics}, volume 452 of {\em London
  Math. Soc. Lecture Note Ser.}, page 289–326. 2018.

\bibitem{yang2017multifractal}
Xiyang~IA Yang and Adri{\'a}n Lozano-Dur{\'a}n.
\newblock A multifractal model for the momentum transfer process in
  wall-bounded flows.
\newblock {\em Journal of fluid mechanics}, 824, 2017.

\end{thebibliography}
\end{document}